%% file: AB_classification_II_Ben_Whale.tex
\documentclass{article}

\usepackage[USenglish]{babel} 
\usepackage[T1]{fontenc}
\usepackage[ansinew]{inputenc}
\usepackage{lmodern}

\usepackage{graphicx} 
\usepackage{color}

\usepackage[margin=1.5in]{geometry}

\expandafter\let\csname equation*\endcsname\relax
\expandafter\let\csname endequation*\endcsname\relax
\usepackage{amsmath}
\usepackage{amsthm}
\usepackage{amsfonts}

\newtheorem{Thm}{Theorem}
\newtheorem{Def}[Thm]{Definition}
\newtheorem{Prop}[Thm]{Proposition}
\newtheorem{Cor}[Thm]{Corollary}
\newtheorem{Lem}[Thm]{Lemma}

\newcommand{\bpp}{b.p.p}
\newcommand{\man}[1]{\ensuremath{\mathcal{#1}}}
\newcommand{\cur}[1]{\ensuremath{\mathcal{#1}}}
\newcommand{\AB}[1][M]{\ensuremath{\mathcal{B}(\man{#1})}}
\newcommand{\covers}{\ensuremath{\rhd}}
\newcommand{\Bound}[1]{\ensuremath{\textnormal{B}({#1})}}
\newcommand{\EmptySet}{\ensuremath{\emptyset}}

\newcommand{\Reg}[1]{\ensuremath{\textnormal{Reg}({#1})}}
\newcommand{\Irreg}[1]{\ensuremath{\textnormal{Irreg}({#1})}}
\newcommand{\Inf}[2]{\ensuremath{\textnormal{Inf}({#1},{#2})}}
\newcommand{\RemInf}[2]{\ensuremath{\textnormal{RemInf}({#1},{#2})}}
\newcommand{\EssInf}[2]{\ensuremath{\textnormal{EssInf}({#1},{#2})}}
\newcommand{\MixInf}[2]{\ensuremath{\textnormal{MixInf}({#1},{#2})}}
\newcommand{\PureInf}[2]{\ensuremath{\textnormal{PureInf}({#1},{#2})}}
\newcommand{\IRem}[1]{\ensuremath{\textnormal{Rem}_{\textnormal{Inf}}({#1})}}
\newcommand{\SRem}[2]{\ensuremath{\textnormal{Rem}_{\textnormal{Sing}}({#1},{#2})}}
\newcommand{\IEss}[1]{\ensuremath{\textnormal{Ess}_{\textnormal{Inf}}({#1})}}
\newcommand{\SEss}[2]{\ensuremath{\textnormal{Ess}_{\textnormal{Sing}}({#1},{#2})}}
\newcommand{\Mix}[1]{\ensuremath{\textnormal{Mix}({#1})}}
\newcommand{\Pure}[1]{\ensuremath{\textnormal{Pure}({#1})}}
\newcommand{\Sing}[2]{\ensuremath{\textnormal{Sing}({#1},{#2})}}
\newcommand{\RemSing}[2]{\ensuremath{\textnormal{RemSing}({#1},{#2})}}
\newcommand{\EssSing}[2]{\ensuremath{\textnormal{EssSing}({#1},{#2})}}
\newcommand{\MixSing}[2]{\ensuremath{\textnormal{MixSing}({#1},{#2})}}
\newcommand{\PureSing}[2]{\ensuremath{\textnormal{PureSing}({#1},{#2})}}
\newcommand{\NSing}[2]{\ensuremath{\textnormal{NonSing}({#1},{#2})}}

\newcommand{\App}[2]{\ensuremath{\textnormal{App}({#1},{#2})}}
\newcommand{\Unapp}[2]{\ensuremath{\textnormal{Nonapp}({#1},{#2})}}
\newcommand{\AppBound}[2]{\ensuremath{\textnormal{App}_{\textnormal{\scriptsize  Sing}}({#1},{#2})}}
\newcommand{\AppUnbound}[2]{\ensuremath{\textnormal{App}_{\textnormal{\scriptsize  Inf}}({#1},{#2})}}
\newcommand{\Env}[1][\man{M}]{\ensuremath{\Phi({#1})}}

\newcommand{\BPPSet}[1][\man{M}]{\ensuremath{\text{BPP}({#1})}}

\newcommand{\ABInf}[1]{\ensuremath{\textsc{Inf}({#1})}}
\newcommand{\ABMixInf}[1]{\ensuremath{\textsc{MixInf}({#1})}}
\newcommand{\ABPureInf}[1]{\ensuremath{\textsc{PureInf}({#1})}}
\newcommand{\ABSing}[1]{\ensuremath{\textsc{Sing}({#1})}}
\newcommand{\ABMixSing}[1]{\ensuremath{\textsc{MixSing}({#1})}}
\newcommand{\ABPureSing}[1]{\ensuremath{\textsc{PureSing}({#1})}}

\newcommand{\ABApp}[1]{\ensuremath{{\textsc{App}}({#1})}}
\newcommand{\ABUnapp}[1]{\ensuremath{{\textsc{Nonapp}}({#1})}}
\newcommand{\ABReg}[1]{\ensuremath{{\textsc{Indet}}({#1})}}

\begin{document}

\title{The dependence of the abstract boundary classification on a set of curves II: How the classification changes when the bounded parameter property satisfying set of curves changes}
\author{B E Whale\footnote{bwhale@maths.otago.ac.nz} \footnote{Department of Mathematics and Statistics, University of Otago}}

\maketitle

\begin{abstract}
	The abstract boundary uses sets of curves with the bounded parameter
  property (b.p.p.) to classify the elements of the abstract boundary
  into regular points, singular points, points at infinity and so on.
  Building on the material of Part one of this two part series,
  we show how this classification changes when the set of b.p.p.
  satisfying curves changes.
\end{abstract}

\section{Introduction}

  A boundary construction in General Relativity
  is a method to attach
  `ideal' points to a Lorentzian manifold. The constructions are 
  designed so that the ideal points can be classified 
  into physically
  motivated classes such as regular points, singular points, points
  at infinity and so on.

  To do this most boundary constructions use, implicitly or explicitly, 
  a set of curves, usually with a particular type of parametrization. 
  For example the $g$-boundary, \cite{Geroch1968a}, relies on incomplete 
  geodesics with affine parameter, 
  the $b$-boundary, \cite{Schmidt1971}, on incomplete curves with generalised 
  affine parameter and the $c$-boundary, \cite{GerochPenroseKronheimer1972},
  and its modern variants,
  \cite{Marolf2003New,citeulike:8211160,Flores2010Final}, on
  endless causal curves. 
  
  Papers such as \cite{ScottSzekeres1994,Geroch1868} 
  reiterate the point that careful consideration of the set of curves
  used in a classification of boundary points is 
  needed to get a correct definition
  of a singularity. Indeed, the issues with giving a consistent physical 
  interpretation, raised by the non-Hausdorff and non-$T_1$ separation
  properties of the $g$-, $b$- and older $c$-boundaries,
  is related to the set of curves used for classification `being too big', 
  e.g.\ including precompact timelike geodesics. For these boundaries, as the
  set of curves is also 
  connected to the construction of the boundary points, the
  inclusion of `too many curves' is part of the root cause of these
  separation properties, 
  \cite{Ashley2002a,Whale2010}. For example the non-Hausdorff behaviour
  of the $b$-boundary is directly related
  to the existence of inextendible incomplete curves that have more than
  one limit point, 
  \cite[Proposition 8.5.1]{HawkingEllis1973}.
  
  The abstract boundary, \cite{ScottSzekeres1994},
  differs from the boundaries mentioned above
  as its construction does not depend on a set of curves. It does use,
  however, a set of curves for the physical classification of its elements.
  This begs the question of what happens to the classification when the
  set of curves changes.
  This paper is the second in a series of two papers which answers this question.

	A set of curves must satisfy the bounded parameter property (b.p.p.) in order
	to be used for the classification of abstract boundary points.
	Unfortunately, a b.p.p.\ satisfying set may contain curves that do not contribute to the
	classification of abstract boundary points. As a consequence the standard algebra
	of sets, $\subset,\cup,\cap$, does not tell the full story regarding the relationships
	between b.p.p.\ satisfying sets, from the point of view of the abstract boundary.
	See Part I, \cite{Whale2012a}, for details and examples of this.
	Part I resolved this issue by generalising $\subset,\cup,\cap$. We use
	this generalisation to describe the relationships between the b.p.p.\ satisfying sets
	that we consider in this paper.
	
	The paper is divided into
	four sections. This section continues with a brief presentation of the
	classification of boundary and abstract boundary points. 
	Section \ref{altdef}
	presents an alternate definition of the classification of boundary points. This allows
	us to reduce our analysis from the study of the 15 sets of the classification
	to the study of 3 sets. This is a substantial simplification.
	Section \ref{sec.boundarychanges} shows how the boundary point classification
	changes when the set of curves changes, while Section \ref{sec.abstratboundarychanges}
	shows how the abstract boundary classification changes.
	
	\subsection{Preliminary Definitions}\label{sec.Background}
		We shall only consider manifolds, \man{M}, that are
		paracompact, Hausdorff, connected, $C^\infty$-manifolds 
		with a metric, $g$.
		Please refer to \cite{Whale2012a} for the definition of a curve and the bounded parameter property (b.p.p.).
		
		\begin{Def}[{\cite[Definition 9]{ScottSzekeres1994}}]
			An embedding, $\phi:\man{M}\to\man{M}_\phi$, of $\man{M}$ is an envelopment 
			if $\man{M}_\phi$ has 
			the same dimension as $\man{M}$.  Let $\Env$ be the set of all envelopments 
			of $\man{M}$.
		\end{Def}
		
		\begin{Def}[{\cite[Definition 4]{ScottSzekeres1994} and \cite[Definition 7]{Whale2012a}}]\label{BPPSet}
			Let $\BPPSet$ be the set of all sets of curves with the \bpp. That is
			$\BPPSet=\{\cur{C}: \cur{C}$ is a set of curves with the \bpp.$\}.$
		\end{Def}
		
		\begin{Def}[{\cite[Definition 14 and 22, Theorem 18]{ScottSzekeres1994}}]
		  Let $\Bound{\man{M}}$ be the set of all ordered pairs $(\phi,U)$
		  of envelopments $\phi$ and subsets $U$ of 
		  $\partial\phi(\man{M})=\overline{\phi(\man{M})}-\phi(\man{M})$. That is,
		  \[
		    \Bound{\man{M}}=\{(\phi,U): \phi\in\Env,\, U\subset\partial\phi(\man{M})\}.
		  \]
		  An element $(\phi,U)$ of $\Bound{\man{M}}$, or just 
		  $U\subset\partial\phi(\man{M})$, is called a boundary set.
		  If $U=\{p\}$ then $(\phi,\{p\})$, or just $p\in\partial\phi(\man{M})$, is called a boundary point.
		  
		  Define a partial order $\covers$ on 
		  $\Bound{\man{M}}$ by $(\phi, U)\covers(\psi, V)$ if and only 
		  if for every sequence $\{x_i\}$ in \man{M}, $\{\psi(x_i)\}$ has a 
		  limit point in $V$ implies that $\{\phi(x_i)\}$ has a limit point in $U$.  
		  We can construct an equivalence relation $\equiv$ on $\Bound{\man{M}}$ by 
		  $(\phi, U)\equiv(\psi, V)$ if and only if $(\phi, U)\covers(\psi, V)$ and
		  $(\psi, V)\covers(\phi, U)$. We denote the equivalence class of $(\phi,U)$ by 
		  $[(\phi, U)]$.
		  
		  The abstract boundary is the set
		  \[
		    \AB=\left\{[(\phi,U)]\in\frac{\Bound{\man{M}}}{\equiv}:\exists 
		    (\psi,\{p\})\in[(\phi,U)] \right\}.
		  \]
		  It is the set of all equivalence classes of $\Bound{\man{M}}$ under the 
		  equivalence relation $\equiv$ that contain an element $(\psi,\{p\})$ where 
		  $p\in\partial\psi(\man{M})$.   The elements of the abstract boundary are referred to as abstract boundary
		  points.
		\end{Def}
		
		The standard algebra of sets, $\subset,\cup,\cap$, does not 
		respect the b.p.p., \cite[Section 2.1]{Whale2012a}. The first part of this series, \cite{Whale2012a},
		addressed this problem by defining a generalisation, $\subset_{\textnormal{b.p.p.}},\cup_{\textnormal{b.p.p.}},\cap_{\textnormal{b.p.p.}}$, of the standard algebra
		of sets over $\BPPSet$ that also behaves well with respect 
		to the classification.
		The details of this generalisation necessary for this paper
		are presented at the beginning of Section \ref{sec.boundarychanges}.
		
		We now present the classification of
		boundary points. In
		\cite{ScottSzekeres1994} the definitions below include references to the 
		differentiability of the metrics involved.  While the differentiability
		of regular points, singular points and points at infinity is an important 
		part of the subject we shall not need this here.  It is an easy matter
		to extend the definitions below and the results of the following sections to 
		include references to the differentiability of the boundary 
		points considered.
		
		\begin{Def}[Regular Boundary Point, 
		  {\cite[Definition 28]{ScottSzekeres1994}}]\label{def.regularpoints}
			A boundary point $p\in\partial\phi(\man{M})$, $\phi\in\Env$, 
			is said to be regular if there exists $\psi\in\Env$ such that
			\begin{enumerate}
				\item $\phi(\man{M})\cup\{p\}\subset\man{M}_\psi$ and $\man M_\psi$ is
					a regular submanifold of $\man{M}_\phi$,
				\item $\psi(x)=\phi(x)$, for all $x\in\man M$, and
				\item there exists a metric $\hat{g}$ on $\man{M}_\psi$ so that 
				$\hat{g}|_{\psi(\man{M})}=g$.
			\end{enumerate}
			
			We make the following definitions;
			\begin{align*}
				\Reg{\phi}&=\{p\in\partial\phi(\man{M}): p\text{ is a }
				  \text{regular point}\}\\
				\Irreg{\phi}&=\{p\in\partial\phi(\man{M}): p
				  \text{ is not a regular point}\}\\
							&=\partial\phi(\man{M})-\Reg{\phi}.
			\end{align*}
		\end{Def}
				
		\begin{Def}[Approachable and Unapproachable Points, 
		  {\cite[Definition 24]{ScottSzekeres1994}}]
			Let $\phi\in\Env$ and $\cur{C}$ be a set of curves with the b.p.p. 
			A boundary point $p\in\partial\phi(\man{M})$ is approachable if 
			there exists $\gamma\in\cur{C}$ so that $p$ is a limit point of the 
			image of the curve $\phi\circ\gamma$.
			
			We make the following definitions;
			\begin{align*}
				\App{\phi}{\cur{C}}&=\{p\in\partial\phi(\man{M}):p
				  \text{ is approachable}\}\\
				\Unapp{\phi}{\cur{C}}&=\{p\in\partial\phi(\man{M}):p
				  \text{ is unapproachable}\}\\
							&=\partial\phi(\man{M})-\App{\phi}{\cur{C}}.
			\end{align*}
		\end{Def}
				
		\begin{Def}[Point at Infinity, 
		  {\cite[Definition 31, 34 and 36]{ScottSzekeres1994}}]\label{def.infpoints}
			Let $\phi\in\Env$ and $\cur{C}$ be a set of curves with the b.p.p.
			A boundary point $p\in\partial\phi(\man{M})$ is said to be a 
			point at infinity if
			\begin{enumerate}
				\item $p\not\in\Reg{\phi}$,
				\item $p\in\App{\phi}{\cur{C}}$,
				\item For all $\gamma\in\cur{C}$, if $p$ is a limit point of 
				$\phi\circ\gamma$ then $\gamma$ is unbounded.
			\end{enumerate}
			
			We make the following definitions;
			\begin{align*}
				\Inf{\phi}{\cur{C}}&=\{p\in\partial\phi(\man{M}): p
				  \text{ is a point at infinity}\}\\
				\RemInf{\phi}{\cur{C}}&=\{p\in \Inf{\phi}{\cur{C}}:\exists(\psi,U)
				  \in \Bound{M}\text{ with}\\ 
					&\quad\quad\quad\quad  U\subset\Reg{\psi}
					\text{ so that}\ (\psi,U)\covers(\phi,\{p\})\}\\
				\EssInf{\phi}{\cur{C}}&=\Inf{\phi}{\cur{C}}-\RemInf{\phi}{\cur{C}}\\	
				\MixInf{\phi}{\cur{C}}&=\{p\in\EssInf{\phi}{\cur{C}}:
				  \exists(\psi,\{q\})\in \Bound{M}\text{ with}\\ 
					&\quad\quad\quad\quad  q\in\Reg{\psi}
					\text{ so that}\ (\phi,\{p\})\covers(\psi,\{q\})\}\\
				\PureInf{\phi}{\cur{C}}&=\EssInf{\phi}{\cur{C}}-\MixInf{\phi}{\cur{C}}
			\end{align*}
			Elements of $\RemInf{\phi}{\cur{C}}$
			are referred to as removable points at
			infinity, $\EssInf{\phi}{\cur{C}}$ as essential points at
			infinity, $\MixInf{\phi}{\cur{C}}$ as mixed points at infinity and 
			$\PureInf{\phi}{\cur{C}}$ as pure points at infinity.
		\end{Def}
		
		\begin{Def}[Singular Point, 
		  {\cite[Definition 37, 40, 41, 42 and 44]{ScottSzekeres1994}}]\label{def.singpoints}
			Let $\phi\in\Env$ and $\cur{C}$ be a set of curves with the \bpp.
			A boundary point $p\in\partial\phi(\man{M})$ is said to be a singularity if
			\begin{enumerate}
				\item $p\not\in\Reg{\phi}$,
				\item $p\in\App{\phi}{\cur{C}}$,
				\item There exists $\gamma\in\cur{C}$ so that $p$ is a limit point of 
				$\phi\circ\gamma$ and $\gamma$ is bounded.
			\end{enumerate}
			
			We can make the following definitions;
			\begin{align*}
				\Sing{\phi}{\cur{C}}&=\{p\in\partial\phi(\man{M}): p
				  \text{ is a singularity}\}\\
				\NSing{\phi}{\cur{C}}&=\partial\phi(\man{M})-\Sing{\phi}{\cur{C}}\\
				\RemSing{\phi}{\cur{C}}&=\{p\in\Sing{\phi}{\cur{C}}:
				  \exists(\psi,U)\in \Bound{M}\text{ with}\\ 
					&\quad\quad\quad\quad  U\subset\NSing{\psi}{\cur{C}}
					\text{ so that}\ (\psi,U)\covers(\phi,\{p\})\}\\
				\EssSing{\phi}{\cur{C}}&=\Sing{\phi}{\cur{C}}-\RemSing{\phi}{\cur{C}}\\
				\MixSing{\phi}{\cur{C}}&=\{p\in\EssSing{\phi}{\cur{C}}:
				  \exists(\psi,\{q\})\in \Bound{M}\text{ with}\\ 
					&\quad\quad\quad\quad  q\in\Reg{\psi}
					\text{ so that}\ (\phi,\{p\})\covers(\psi,\{q\})\}\\
				\PureSing{\phi}{\cur{C}}&=\EssSing{\phi}{\cur{C}}-\MixSing{\phi}{\cur{C}}
			\end{align*}		
			Elements of $\RemSing{\phi}{\cur{C}}$ are called removable singularities, 
			$\EssSing{\phi}{\cur{C}}$ are
			called essential singularities, $\MixSing{\phi}{\cur{C}}$ are called 
			mixed (or directional) singularities, $\PureSing{\phi}{\cur{C}}$ are
			called pure singularities.
		\end{Def}
		
		In \cite{ScottSzekeres1994} the properties of the above definitions 
		are explored with respect to
		the equivalence relation $\equiv$.  Scott and Szekeres
		show that the following definitions are well defined.  
		
		\begin{Def}[Approachable and unapproachable abstract boundary points, 
		  {\cite[Section 5]{ScottSzekeres1994}}]\label{App&UnAppABpoints}
			Let $\cur{C}$ be a set of curves with the b.p.p.\ then we can define
			\begin{align*}
				\ABApp{\cur{C}}&=\{[(\phi,\{p\})]\in\AB: p\in\App{\phi}{\cur{C}}\}\\
				\ABUnapp{\cur{C}}&=\{[(\phi,\{p\})]\in\AB:
						p\in\Unapp{\phi}{\cur{C}}\}
			\end{align*}
			Elements of $\ABApp{\cur{C}}$ are called approachable abstract boundary points 
			and elements of
			$\ABUnapp{\cur{C}}$ are called unapproachable abstract boundary points.
		\end{Def}
				
		\begin{Def}[Indeterminate abstract boundary points, 
		  {\cite[Section 5]{ScottSzekeres1994}}]\label{IndABPOints}
			Let $\cur{C}$ be a set of curves with the b.p.p., then an abstract boundary 
			point $[(\phi,\{p\})]\in\AB$ 
			is an indeterminate abstract boundary point if one of the following is true,
			\begin{enumerate}
				\item $p\in\Reg{\phi}$,
				\item $p\in\RemInf{\phi}{\cur{C}}$, or
				\item $p\in\RemSing{\phi}{\cur{C}}$.
			\end{enumerate}
			Let,
			\begin{multline*}
				\ABReg{\cur{C}}=\{[(\phi,\{p\})]\in\AB: p\text{ is an }
						\text{indeterminate abstract boundary point}\}
			\end{multline*}
		\end{Def}
		
		\begin{Def}[Abstract boundary points at infinity, 
		  {\cite[Definition 47 and the paragraph after Definition 48]{ScottSzekeres1994}}]\label{InfAbPoints}
			Let $\cur{C}$ be a set of curves with the b.p.p., then we can 
			make the following definitions
			\begin{align*}
				\ABInf{\cur{C}}&=\{[(\phi,\{p\})]\in\AB: 
				  p\in\EssInf{\phi}{\cur{C}}\}\\
				\ABMixInf{\cur{C}}&=\{[(\phi,\{p\})]\in\AB: 
				  p\in\MixInf{\phi}{\cur{C}}\}\\
				\ABPureInf{\cur{C}}&=\{[(\phi,\{p\})]\in\AB: 
				  p\in\PureInf{\phi}{\cur{C}}\}.
			\end{align*}
			Elements of $\ABInf{\cur{C}}$ are called abstract boundary points at
			infinity, elements of $\ABMixInf{\cur{C}}$ are called abstract boundary
			mixed points at infinity and elements of $\ABPureInf{\cur{C}}$
			are called abstract boundary pure points at infinity.
		\end{Def}
		
		\begin{Def}[Singular abstract boundary points, 
		  {\cite[Definition 48 and the following paragraph]{ScottSzekeres1994}}]\label{SingAbPoints}
			Let $\cur{C}$ be a set of curves with the \bpp, then we can make the following definitions
			\begin{align*}
				\ABSing{\cur{C}}&=\{[(\phi,\{p\})]\in\AB: p\in\EssSing{\phi}{\cur{C}}\}\\
				\ABMixSing{\cur{C}}&=\{[(\phi,\{p\})]\in\AB: p\in\MixSing{\phi}{\cur{C}}\}\\
				\ABPureSing{\cur{C}}&=\{[(\phi,\{p\})]\in\AB: p\in\PureSing{\phi}{\cur{C}}\}.
			\end{align*}
			Elements of $\ABSing{\cur{C}}$ are called abstract boundary singular
			points, elements of $\ABMixSing{\cur{C}}$ are called abstract boundary
			mixed (or directional) singularities and elements of 
			$\ABPureSing{\cur{C}}$ are called abstract boundary pure
			singularities.
		\end{Def}

\section{Alternate definitions of the classes of boundary points}\label{altdef}

  Before we study how the classification given above changes with respect to 
  changes in the b.p.p.\ satisfying set of curves, we revisit the 
  definitions of the classes.
  We generalise a few of the concepts behind the classification 
  presented in Section \ref{sec.Background}
  and express each of the classes of the classification
  as a union / intersection of more `primitive' sets.  
  This will allow us to reduce the study of the 15 sets of the boundary point 
  classification 
  to the study of 3 sets.

  In \cite[Definition 6]{Whale2012a} the following subdivision of
  $\App{\phi}{\cur{C}}$, $\phi\in\Env$, $\cur{C}\in\BPPSet$ was introduced.
  \begin{Def}[{\cite[Definition 6]{Whale2012a}}]
		Let $\AppBound{\phi}{\cur{C}}$ and $\AppUnbound{\phi}{\cur{C}}$ be defined by, 
		\begin{multline*}
			\AppBound{\phi}{\cur{C}}=\{p\in\App{\phi}{\cur{C}}:\ \text{there exists}\ \gamma\in\cur{C}\text{, bounded,}\\ \text{so that}\ 
			p\ \text{is a limit point of}\ \phi\circ\gamma\}
		\end{multline*}
		\begin{multline*}
			\AppUnbound{\phi}{\cur{C}}=\{p\in\App{\phi}{\cur{C}}:\ \text{for all}\ \gamma\in\cur{C},
			\text{if}\ p\ \text{is a}\\  \text{limit point of}\ 
			\phi\circ\gamma\ \text{then}\ \gamma\ \text{has unbounded 
			parameter}\}.
		\end{multline*}
  \end{Def}
  This generalises the idea of `singular point' and `point at infinity'.
  We can generalise the concepts of `mixed' and `pure' points.
  \begin{Def}\label{def.mixedpurepoints}
		Let $\phi\in\Env$. A boundary point $p\in\partial\phi(\man{M})$ is mixed if there 
		exists $(\psi,\{q\})\in\Bound{\man{M}}$ such that
		\begin{enumerate}
						\item $(\phi,\{p\})\covers(\psi,\{q\})$,
						\item $q\in\Reg{\psi}$.
		\end{enumerate}
		If $p$ is not mixed then we shall say that it is a pure boundary point.
		
		We can make the following definitions;
		\begin{align*}
						\Mix{\phi}&=\{p\in\partial\phi(\man{M}): p\text{ is a mixed boundary point}\}\\
						\Pure{\phi}&=\partial\phi(\man{M})-\Mix{\phi}.
		\end{align*}
  \end{Def}
  
	As will become clear, the analysis of how the classification changes would be much easier if 
	we could also define a `removable' point independently of points at
	infinity and singular points. Unfortunately the
  small difference in the definition of $\RemInf{\phi}{\cur{C}}$ and 
  $\RemSing{\phi}{\cur{C}}$ is a serious (and probably fatal) impediment to 
  this.  If
  $p\in\RemInf{\phi}{\cur{C}}$ then 
  there must exist $(\psi,U)\in\Bound{\man{M}}$ so that $U\subset\Reg{\psi}$  
  and 
  $(\psi,U)\covers (\phi,\{p\})$. If
  $p\in\RemSing{\phi}{\cur{C}}$ then
  there must exist $(\psi,U)\in\Bound{\man{M}}$ so that
  $U\subset\NSing{\phi}{\cur{C}}$ and $(\psi,U)\covers (\phi,\{p\})$.
  Thus the definition of a removable point at infinity depends on $\Reg{\phi}$,
  while the definition of a removable singularity depends on 
  $\NSing{\phi}{\cur{C}}$.  
  Hence, while the two definitions are similar, to provide a single
  definition of a removable point we would need to show something like
  $p\in\RemSing{\phi}{\cur{C}}$ if and only if 
  there exists 
  $(\psi,U)\in\Bound{\man{M}}$ so that $U\subset\Reg{\phi}$ and 
  $(\psi,U)\covers (\phi,\{p\})$.
  Such a result is,  almost certainly, false due to the existence of
  non-approachable boundary points, see the proof of
  \cite[Theorem 43]{ScottSzekeres1994}.
  
  We avoid this issue by differentiating 
  between the two `types' of removable point.
  \begin{Def}\label{def.reminfpoints}
		Let $\IRem{\phi}$ be defined as
		\begin{multline*}
			\IRem{\phi}=\{p\in\partial\phi(\man{M}):\exists (\psi,U)\in\Bound{\man{M}}\text{ so that }\\ U\subset\Reg{\phi}\text{ and }(\psi,U)\covers(\phi,\{p\})\}.  
		\end{multline*}
		This is the set of all points that are removable in the sense of the definition of a removable point at infinity.
		Let $\IEss{\phi}=\partial\phi(\man{M})-\IRem{\phi}$.
  \end{Def}
          
  \begin{Def}\label{def.srem-sess}
		Let $\SRem{\phi}{\cur{C}}$ be defined as,
		\begin{multline*}
			\SRem{\phi}{\cur{C}}=\{p\in\partial\phi(\man{M}):\exists (\psi,U)\in\Bound{\man{M}}\\
			\text{ so that }U\subset\NSing{\phi}{\cur{C}}\text{ and }(\psi,U)\covers(\phi,\{p\})\}
		\end{multline*}
		This is the set of all points that are removable in the sense of the definition of a removable singularity.
		Let $\SEss{\phi}{\cur{C}}=\partial\phi(\man{M})-\SRem{\phi}{\cur{C}}.$
  \end{Def}
  Since $\Reg{\phi}\subset \NSing{\phi}{\cur{C}}$ we see that
  $\IRem{\phi}\subset\SRem{\phi}{\cur{C}}$, for all $\phi\in\Env$ and
  $\cur{C}\in\BPPSet$.
   
  These definitions will eventually lead to some interesting results.  
  In particular, when expanding a b.p.p.\ satisfying set to include 
  more curves
  it is possible for a pure point at infinity to become a removable singularity.
 
  The classes defined in Section \ref{sec.Background}
  can be expressed in terms of the sets we have just given.
  \begin{Prop}\label{Prop.Altdefinition}
		Let $\phi\in\Env$ and $\cur{C}$ be a set of curves with the \bpp. 
		Then we have that,
		\begin{align*}
			\Inf{\phi}{\cur{C}}&=\Irreg{\phi}\cap\AppUnbound{\phi}{\cur{C}}\\
			\RemInf{\phi}{\cur{C}}&=\Irreg{\phi}\cap\AppUnbound{\phi}{\cur{C}}\cap\IRem{\phi}\\
			\EssInf{\phi}{\cur{C}}&=\Irreg{\phi}\cap\AppUnbound{\phi}{\cur{C}}\cap\IEss{\phi}\\	
			\MixInf{\phi}{\cur{C}}&=
							\Irreg{\phi}\cap\AppUnbound{\phi}{\cur{C}}\cap\IEss{\phi}\cap\Mix{\phi}\\
			\PureInf{\phi}{\cur{C}}&=
			\Irreg{\phi}\cap\AppUnbound{\phi}{\cur{C}}\cap\IEss{\phi}\cap\Pure{\phi}\\ \\
			\Sing{\phi}{\cur{C}}&=\Irreg{\phi}\cap\AppBound{\phi}{\cur{C}}\\
			\NSing{\phi}{\cur{C}}&=\Reg{\phi}\cup\AppUnbound{\phi}{\cur{C}}\cup\Unapp{\phi}{\cur{C}}\\
			\RemSing{\phi}{\cur{C}}&=\Irreg{\phi}\cap\AppBound{\phi}{\cur{C}}\cap\SRem{\phi}{\cur{C}}\\
			\EssSing{\phi}{\cur{C}}&=\Irreg{\phi}\cap\AppBound{\phi}{\cur{C}}\cap\SEss{\phi}{\cur{C}}\\
			\MixSing{\phi}{\cur{C}}&=
			\Irreg{\phi}\cap\AppBound{\phi}{\cur{C}}\cap\SEss{\phi}{\cur{C}}\cap\Mix{\phi}\\
			\PureSing{\phi}{\cur{C}}&=
			\Irreg{\phi}\cap\AppBound{\phi}{\cur{C}}\cap\SEss{\phi}{\cur{C}}\cap\Pure{\phi}.		
		\end{align*}
  \end{Prop}
  \begin{proof}
		The proofs of these statements follow immediately from the
		definitions, except for the equation 
		$\NSing{\phi}{\cur{C}}=\Reg{\phi}\cup\AppUnbound{\phi}{\cur{C}}\cup\Unapp{\phi}{\cur{C}}$
		which requires use of the equation $A-(B\cap C)=(A- B)\cup(A- C)$.
  \end{proof}
  
  As $\AppUnbound{\phi}{\cur{C}}=\App{\phi}{\cur{C}}-\AppBound{\phi}{\cur{C}}$
  and $\SEss{\phi}{\cur{C}}=\partial\phi(\man{M})-\SRem{\phi}{\cur{C}}$,
  Proposition \ref{Prop.Altdefinition} implies that we need only study
  how the three sets $\App{\phi}{\cur{C}}$, $\AppBound{\phi}{\cur{C}}$ and
  $\SRem{\phi}{\cur{C}}$ change under changes of $\cur{C}$
  in order to work out how each of the 15 sets of the boundary point classification change.

\section{How the boundary classification changes as we change \ensuremath{\cur{C}}}\label{sec.boundarychanges}

	It was claimed in the introduction, Section \ref{sec.Background}, that the algebra of sets $\subset_{\textnormal{b.p.p.}}$, $\cup_{\textnormal{b.p.p.}}$, $\cap_{\textnormal{b.p.p.}}$,
	defined on $\BPPSet$, behaves nicely with respect to the classification.
	By this we mean the following \cite[Section 4]{Whale2012a}, if $\cur{C},\cur{D}\in\BPPSet$ and $\phi\in\Env$
	then
	\begin{align*}
	  \cur{C}\subset_{\textnormal{b.p.p.}}\cur{D} \Rightarrow&\left\{\begin{aligned}
	                                               \App{\phi}{\cur{C}}&\subset\App{\phi}{\cur{D}}\\
	                                               \AppBound{\phi}{\cur{C}}&\subset\AppBound{\phi}{\cur{D}}
	                                             \end{aligned}\right.\\[7pt]
		\App{\phi}{\cur{C}},\, \App{\phi}{\cur{D}}&\subset\App{\phi}{\cur{C}\cup_{\textnormal{b.p.p.}}\cur{D}}\\
		\AppBound{\phi}{\cur{C}},\, \AppBound{\phi}{\cur{D}}&\subset\AppBound{\phi}{\cur{C}\cup_{\textnormal{b.p.p.}}\cur{D}}\\[7pt]
		\App{\phi}{\cur{C}\cap_{\textnormal{b.p.p.}}\cur{D}}&\subset \App{\phi}{\cur{C}},\, \App{\phi}{\cur{D}}\\
		\AppBound{\phi}{\cur{C}\cap_{\textnormal{b.p.p.}}\cur{D}}&\subset \AppBound{\phi}{\cur{C}},\, \AppBound{\phi}{\cur{D}}.
	\end{align*}
	Note that $\cur{C}\cup_{\textnormal{b.p.p.}}\cur{D}$ does not always exist due to the 
	possible existence of curves in $\cur{C}$ and
	$\cur{D}$ that have the same image where one is bounded and the other unbounded. See \cite[Section 2.1]{Whale2012a} for details of this.
	The equations, given above, involving $\cur{C}\cup_{\textnormal{b.p.p.}}\cur{D}$ are only valid when $\cur{C}\cup_{\textnormal{b.p.p.}}\cur{D}$ 
	exists.

	Below it is shown how the boundary point classification differs between;
	\begin{itemize}
	  \item $\cur{C}$ and $\cur{D}$ when $\cur{C}\subset_{\textnormal{b.p.p.}}\cur{D}$
	  \item $\cur{C},\cur{D}$ and $\cur{C}\cup_{\textnormal{b.p.p.}}\cur{D}$, when $\cur{C}\cup_{\textnormal{b.p.p.}}\cur{D}$
			can be defined.
	\end{itemize}
	We do not show how the classification differs between $\cur{C},\cur{D}$
	and $\cur{C}\cap_{\textnormal{b.p.p.}}\cur{D}$ as it is
	very unlikely that one would want to restrict the
	set of curves used for analysis of the boundary. Moreover the details
	for this case can be determined by following the pattern of results
	established by the cases $\cur{C}\subset_{\textnormal{b.p.p.}}\cur{D}$
	and $\cur{C}\cup_{\textnormal{b.p.p.}}\cur{D}$.
		
	\subsection{Subset}\label{subsetsClass}
	
		Given $\cur{C}\in\BPPSet$ we investigate how the 
		boundary classification induced by $\cur{C}$
		relates to the boundary classification induced by $\cur{D}\in\BPPSet$
		when $\cur{C}\subset_{b.p.p}\cur{D}$. 
		Because of the results of Section 
		\ref{altdef}, we first focus on the sets $\App{\phi}{\cur{C}}$,
		$\AppBound{\phi}{\cur{C}}$
		and $\SRem{\phi}{\cur{C}}$.
		
		\begin{Prop}\label{prop.InfChanges}
			Let $\cur{C},\cur{D}\in\BPPSet$ so that 
			$\cur{C}\subset_{\textnormal{b.p.p.}}\cur{D}$ then for all $\phi\in\Env$,
			\begin{gather*}
				\App{\phi}{\cur{C}}\subset\App{\phi}{\cur{D}},\\
				\AppBound{\phi}{\cur{C}}\subset\AppBound{\phi}{\cur{D}},\\
				\SRem{\phi}{\cur{D}}\subset\SRem{\phi}{\cur{C}}.
			\end{gather*}
		\end{Prop}
		\begin{proof}

			The first two statements follow from Proposition 30 of 
			\cite{Whale2012a}.
			
			From the second statement we know that
			$\partial\phi(\man{M})-\AppBound{\phi}{\cur{D}}\subset\partial\phi(\man{M})-\AppBound{\phi}{\cur{C}}$.
			Thus, from Definition \ref{def.singpoints} and as $A-(B\cap C)=(A-B)\cup(A-C)$, we have that, for all $\phi\in\Env$,
			\begin{align*}
				\NSing{\phi}{\cur{D}}&=\left(\partial\phi(\man{M})-\Irreg{\phi}\right)\cup
						\left(\partial\phi(\man{M})-\AppBound{\phi}{\cur{D}}\right)\\
					&\subset\left(\partial\phi(\man{M})-\Irreg{\phi}\right)\cup
						\left(\partial\phi(\man{M})-\AppBound{\phi}{\cur{C}}\right)\\
					&=\NSing{\phi}{\cur{C}}.
			\end{align*}
			
			Let $p\in\SRem{\phi}{\cur{D}}$ then there exists $(\psi,U)\in\Bound{\man M}$
			so that $U\subset\NSing{\psi}{\cur{D}}$ and $(\psi,U)\covers(\phi,\{p\})$.
			From above we know that $U\subset\NSing{\psi}{\cur{C}}$
			and therefore $p\in\SRem{\phi}{\cur{C}}$, as required.
		\end{proof}

		Using Proposition \ref{prop.InfChanges} we can determine
		how the other sets with a dependence on $\cur{C}$ appearing in the
		left hand side of the equations of Proposition \ref{Prop.Altdefinition}
		are affected.

		\begin{Cor}\label{cor.InfChanges}
			Let $\cur{C},\cur{D}\in\BPPSet$ so that $\cur{C}\subset_{\textnormal{b.p.p.}}\cur{D}$
			then for all $\phi\in\Env$,
			\begin{enumerate}
				\item $\App{\phi}{\cur{D}}-\App{\phi}{\cur{C}}=\App{\phi}{\cur{D}}\cap\Unapp{\phi}{\cur{C}}$,\label{cor.infchanges.1}
				\item  $\AppBound{\phi}{\cur{D}}-\AppBound{\phi}{\cur{C}}=\AppBound{\phi}{\cur{D}}\cap
							\bigl(\Unapp{\phi}{\cur{C}}\cup\AppUnbound{\phi}{\cur{C}}\bigr)$,\label{cor.infchanges.2}
				\item $\AppUnbound{\phi}{\cur{D}}\cap\App{\phi}{\cur{C}}\subset\AppUnbound{\phi}{\cur{C}}$,\label{cor.infchanges.3}
				\item $\AppUnbound{\phi}{\cur{C}}-\bigl(\AppUnbound{\phi}{\cur{D}}\cap\App{\phi}{\cur{C}}\bigr)=
							\AppUnbound{\phi}{\cur{C}}-\AppUnbound{\phi}{\cur{D}}$\newline\vspace{0cm}
							$\hspace{2.5cm}=\AppUnbound{\phi}{\cur{C}}\cap\AppBound{\phi}{\cur{D}}$,\label{cor.infchanges.4}
				\item If $x\in\SRem{\phi}{\cur{C}}-\SRem{\phi}{\cur{D}}$ then for all $(\psi,U)\in\Bound{\man{M}}$ so that $(\psi,U)\covers(\phi,\{x\})$ we have that
							$U\cap\Irreg{\phi}\cap\bigl(\AppBound{\psi}{\cur{D}}-\AppBound{\psi}{\cur{C}}\bigr)\neq\EmptySet$,\label{cor.infchanges.5}
				\item $\SEss{\phi}{\cur{C}}\subset\SEss{\phi}{\cur{D}}$,\label{cor.infchanges.6}
				\item $\SEss{\phi}{\cur{D}}-\SEss{\phi}{\cur{C}}=\SRem{\phi}{\cur{C}}-\SRem{\phi}{\cur{D}}$.\label{cor.infchanges.7}
			\end{enumerate}
		\end{Cor}
		\begin{proof}
			The proofs of \eqref{cor.infchanges.1}, \eqref{cor.infchanges.2}, \eqref{cor.infchanges.3}, \eqref{cor.infchanges.4} and \eqref{cor.infchanges.6}  
			follow directly from Proposition
			\ref{prop.InfChanges} and the definitions. We include them here for completeness.
			The proof of item \eqref{cor.infchanges.5} follows from Proposition \ref{Prop.Altdefinition}
			and the definition of $\SRem{\phi}{\cur{C}}$.
			The proof of item \eqref{cor.infchanges.7} follows from the definition of
			$\SEss{\phi}{\cur{C}}$ and the standard set relations
			$A-(B- C)=(A\cap C)\cup(A- B)$
			and $(A- B)\cap C=A\cap(C- B)$.
		\end{proof}
		
		Figures \ref{figure appinf} and \ref{figure remsing}
		give a graphic representation of Proposition \ref{prop.InfChanges}
		and Corollary \ref{cor.InfChanges}.
		Proposition \ref{prop.InfChanges} and Corollary \ref{cor.InfChanges}
		now let us determine how the boundary point classification itself
		changes.
		
%		 {\smash{\ensuremath{\AppUnbound{\phi}{\cur{D}}}}}}%
%    {\smash{\ensuremath{\AppBound{\phi}{\cur{D}}}}}}%
%    {\smash{\ensuremath{\AppUnbound{\phi}{\cur{C}}}}}}%
%    {\smash{\ensuremath{\AppBound{\phi}{\cur{C}}}}}}%
%    {\smash{\ensuremath{\Irreg{\phi}}}}}%
%    {\smash{\ensuremath{\partial\phi(\man{M})}}}}%

		\begin{figure*}
		  \centering
		  \def\svgwidth{\columnwidth}
		  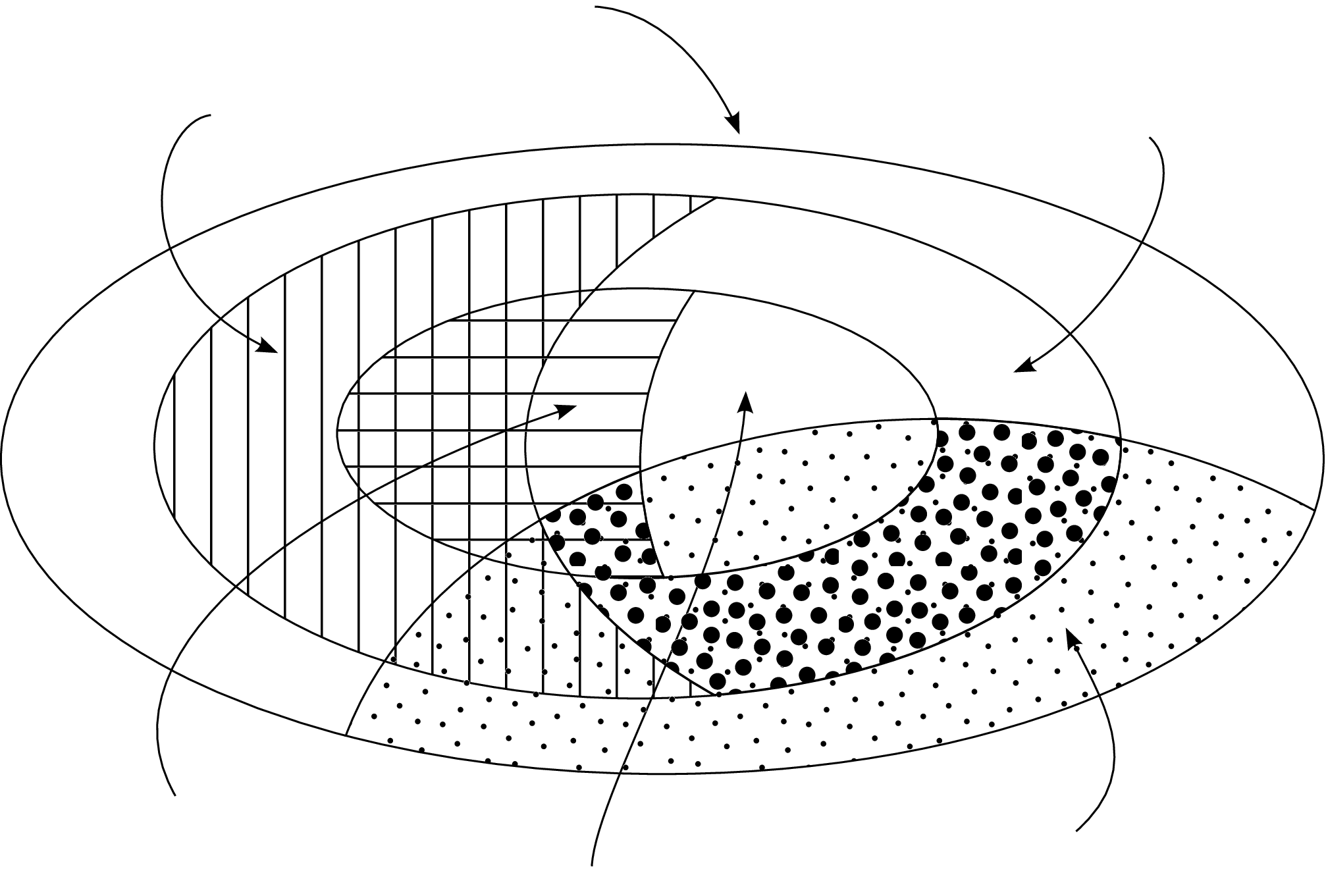
		  \caption{A venn diagram of $\partial\phi(\man{M})$, $\Reg{\phi}$,
		    $\Irreg{\phi}$, $\App{\phi}{\cur{C}}$,
		    $\AppBound{\phi}{\cur{C}}$, $\AppUnbound{\phi}{\cur{C}}$,
		    $\App{\phi}{\cur{D}}$, $\AppBound{\phi}{\cur{D}}$ and 
		    $\AppUnbound{\phi}{\cur{D}}$.
		    The outer oval is $\partial\phi(\man{M})$. The dotted region,
		    both small and large dots,
		    is $\Irreg{\phi}$.		  
		    The middle oval is $\App{\phi}{\cur{D}}$, the inner oval is 
		    $\App{\phi}{\cur{C}}$. The region ruled by vertical lines is 
		    $\AppUnbound{\phi}{\cur{D}}$, the region of the middle oval
		    not ruled by vertical lines is $\AppBound{\phi}{\cur{D}}$. The region
		    ruled by horizontal lines is $\AppUnbound{\phi}{\cur{C}}$, the region 
		    of the
		    inner oval not
		    ruled by horizontal lines is $\AppBound{\phi}{\cur{C}}$. 
		    The region ruled by horizontal lines and not ruled by vertical lines
		    is described by item \eqref{cor.infchanges.4} of Corollary 
		    \ref{cor.InfChanges}. The region covered by big dots is
		    $\Irreg{\phi}\cap\bigl(\AppBound{\phi}{\cur{D}}-\AppBound{\phi}{\cur{C}}
		    \bigr)$, see item \eqref{cor.infchanges.5} of Corollary 
		    \ref{cor.InfChanges}.}\label{figure appinf}
		\end{figure*}
		
% 		{\smash{\ensuremath{\partial\phi(\man{M})}}}}%
%     {\smash{\ensuremath{\SRem{\phi}{\cur{D}}}}}}%
%     {\smash{\ensuremath{\SRem{\phi}{\cur{C}}}}}}%
		
		\begin{figure*}
		  \centering
		  \def\svgwidth{\columnwidth}
		  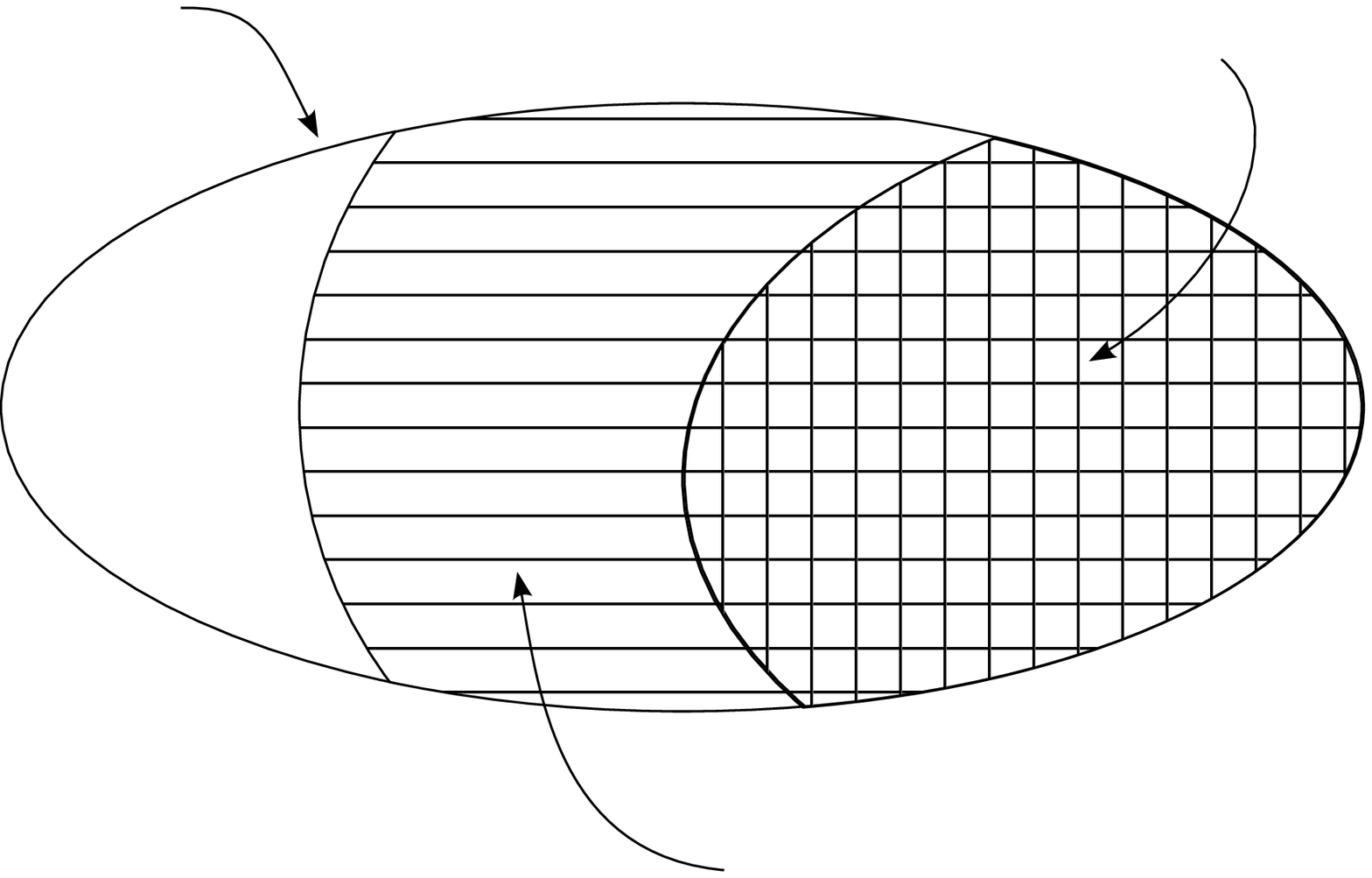
		  \caption{A venn diagram of $\partial\phi(\man{M}),\SRem{\phi}{\cur{C}},\SEss{\phi}{\cur{C}},\SRem{\phi}{\cur{D}}$ and $\SEss{\phi}{\cur{D}}$.
		  The oval is $\partial\phi(\man{M})$. The region ruled by horizontal lines is $\SRem{\phi}{\cur{C}}$, the region not ruled by horizontal lines
		  is $\SEss{\phi}{\cur{C}}$. The region ruled by vertical lines is $\SRem{\phi}{\cur{D}}$, the region not ruled by
		  vertical lines is $\SEss{\phi}{\cur{D}}$. The region of the oval that is ruled by horizontal lines and not ruled by vertical lines
		  is described by item \eqref{cor.infchanges.5} of Corollary \ref{cor.InfChanges}, see Figure \ref{figure appinf}.}\label{figure remsing}
		\end{figure*}

		\begin{Cor}\label{cor.bclasschange}
			Let $\cur{C},\cur{D}\in\BPPSet$ so that $\cur{C}\subset_{\textnormal{b.p.p.}}\cur{D}$ then for all $\phi\in\Env$,
				\begin{enumerate}
					\item  $\Inf{\phi}{\cur{D}}\cap\App{\phi}{\cur{C}}\subset\Inf{\phi}{\cur{C}}.$\vspace{4pt}\newline
									$\Inf{\phi}{\cur{C}}-\bigl(\Inf{\phi}{\cur{D}}\cap\App{\phi}{\cur{C}}\bigr)=\Inf{\phi}{\cur{C}}-\Inf{\phi}{\cur{D}}$\newline\vspace{0cm}
									$\hspace{60pt}=\Irreg{\phi}\cap
									\bigl(\AppUnbound{\phi}{\cur{C}}-\AppUnbound{\phi}{\cur{D}}\bigr).$\vspace{0pt}\newline\label{cor.bclasschange.1}
					\item  $\RemInf{\phi}{\cur{D}}\cap\App{\phi}{\cur{C}}\subset\RemInf{\phi}{\cur{C}}.$\vspace{4pt}\newline
									$\RemInf{\phi}{\cur{C}}-\bigl(\RemInf{\phi}{\cur{D}}\cap\App{\phi}{\cur{C}}\bigr)=\RemInf{\phi}{\cur{C}}-\RemInf{\phi}{\cur{D}}$\newline\vspace{0cm}
									$\hspace{60pt}=\Irreg{\phi}\cap\IRem{\phi}\cap\bigl(\AppUnbound{\phi}{\cur{C}}-\AppUnbound{\phi}{\cur{D}}\bigr).$\vspace{0pt}\newline
									\label{cor.bclasschange.2}
					\item  $\EssInf{\phi}{\cur{D}}\cap\App{\phi}{\cur{C}}\subset\EssInf{\phi}{\cur{C}}.$\vspace{4pt}\newline
									$\EssInf{\phi}{\cur{C}}-\bigl(\EssInf{\phi}{\cur{D}}\cap\App{\phi}{\cur{C}}\bigr)=\EssInf{\phi}{\cur{C}}-\EssInf{\phi}{\cur{D}}$\newline\vspace{0cm}
									$\hspace{60pt}=\Irreg{\phi}\cap\IEss{\phi}\cap\bigl(\AppUnbound{\phi}{\cur{C}}-\AppUnbound{\phi}{\cur{D}}\bigr).$\vspace{0pt}\newline
									\label{cor.bclasschange.3}
					\item  $\MixInf{\phi}{\cur{D}}\cap\App{\phi}{\cur{C}}\subset\MixInf{\phi}{\cur{C}}.$\vspace{4pt}\newline
									$\MixInf{\phi}{\cur{C}}-\bigl(\MixInf{\phi}{\cur{D}}\cap\App{\phi}{\cur{C}}\bigr)=\MixInf{\phi}{\cur{C}}-\MixInf{\phi}{\cur{D}}$\newline\vspace{0cm}
									$\hspace{40pt}=
										\Irreg{\phi}\cap\IEss{\phi}\cap\Mix{\phi}\cap\bigl(\AppUnbound{\phi}{\cur{C}}-\AppUnbound{\phi}{\cur{D}}\bigr).$\vspace{0pt}\newline
									\label{cor.bclasschange.4}
					\item  $\PureInf{\phi}{\cur{D}}\cap\App{\phi}{\cur{C}}\subset\PureInf{\phi}{\cur{C}}.$\vspace{4pt}\newline
									$\PureInf{\phi}{\cur{C}}-\bigl(\PureInf{\phi}{\cur{D}}\cap\App{\phi}{\cur{C}}\bigr)=\PureInf{\phi}{\cur{C}}-\PureInf{\phi}{\cur{D}}$\newline\vspace{0cm}
									$\hspace{40pt}=
										\Irreg{\phi}\cap\IEss{\phi}\cap\Pure{\phi}\cap\bigl(\AppUnbound{\phi}{\cur{C}}-\AppUnbound{\phi}{\cur{D}}\bigr).$\vspace{0pt}\newline
									\label{cor.bclasschange.5}
					\item  $\Sing{\phi}{\cur{C}}\subset\Sing{\phi}{\cur{D}}$\vspace{4pt}\newline
									$\Sing{\phi}{\cur{D}}-\Sing{\phi}{\cur{C}}=\Irreg{\phi}\cap\left(\AppBound{\phi}{\cur{D}}-\AppBound{\phi}{\cur{C}}\right).$\vspace{0pt}\newline
									\label{cor.bclasschange.6}
					\item  $\NSing{\phi}{\cur{D}}\subset\NSing{\phi}{\cur{C}}.$\vspace{4pt}\newline
									$\NSing{\phi}{\cur{C}}-\NSing{\phi}{\cur{D}}
									=\Sing{\phi}{\cur{D}}-\Sing{\phi}{\cur{D}}.$\vspace{0pt}\newline\label{cor.bclasschange.7}
					\item  $\EssSing{\phi}{\cur{C}}\subset\EssSing{\phi}{\cur{D}}.$ \vspace{4pt}\newline
									$\EssSing{\phi}{\cur{D}}-\EssSing{\phi}{\cur{C}}=\EssSing{\phi}{\cur{D}}\cap$\newline\vspace{0cm}
									$\hspace{20pt}\Bigl(\bigl(\AppBound{\phi}{\cur{D}}-\AppBound{\phi}{\cur{C}}\bigr)\bigcup
									\bigl(\SEss{\phi}{\cur{D}}-\SEss{\phi}{\cur{C}}\bigr)\Bigr)$.\vspace{0pt}\newline\label{cor.bclasschange.8}
					\item  $\MixSing{\phi}{\cur{C}}\subset\MixSing{\phi}{\cur{D}}$ \vspace{4pt}\newline
									$\MixSing{\phi}{\cur{D}}-\MixSing{\phi}{\cur{C}}$\newline\vspace{0cm}
									$\hspace{40pt}=\Mix{\phi}\cap\bigl(\EssSing{\phi}{\cur{D}}-\EssSing{\phi}{\cur{C}}\bigr).$\vspace{0pt}\newline\label{cor.bclasschange.9}
					\item  $\PureSing{\phi}{\cur{C}}\subset\PureSing{\phi}{\cur{D}}$\vspace{4pt}\newline
									$\PureSing{\phi}{\cur{D}}-\PureSing{\phi}{\cur{C}}$\newline\vspace{0cm}
									$\hspace{40pt}=\Pure{\phi}\cap\bigl(\EssSing{\phi}{\cur{D}}-\EssSing{\phi}{\cur{C}}\bigr).$\vspace{0pt}\newline\label{cor.bclasschange.10}
					\item  $\RemSing{\phi}{\cur{D}}\cap\Sing{\phi}{\cur{C}}\subset\RemSing{\phi}{\cur{C}}$\vspace{4pt}\newline
									$\RemSing{\phi}{\cur{D}}\cap\RemSing{\phi}{\cur{C}}$\newline\vspace{0cm}
									$\hspace{40pt}=\Irreg{\phi}\cap\AppBound{\phi}{\cur{C}}\cap\SRem{\phi}{\cur{D}}.$\vspace{4pt}\newline
									$\RemSing{\phi}{\cur{C}}-\bigl(\RemSing{\phi}{\cur{D}}\cap\Sing{\phi}{\cur{C}}\bigr)=\RemSing{\phi}{\cur{C}}-\RemSing{\phi}{\cur{D}}$\newline\vspace{0cm}
									$\hspace{40pt}=\Irreg{\phi}\cap\AppBound{\phi}{\cur{C}}\cap\bigl(\SRem{\phi}{\cur{C}}-\SRem{\phi}{\cur{D}}\bigr).$\vspace{4pt}\newline
									$\RemSing{\phi}{\cur{D}}-\RemSing{\phi}{\cur{C}}$\newline\vspace{0cm}
									$\hspace{40pt}=\Irreg{\phi}\cap\SRem{\phi}{\cur{D}}\cap\left(\AppBound{\phi}{\cur{D}}-\AppBound{\phi}{\cur{C}}\right).$\vspace{0pt}\newline
									\label{cor.bclasschange.11}
				\end{enumerate}
						Membership of $\AppUnbound{\phi}{\cur{C}}-\AppUnbound{\phi}{\cur{D}}$,
		$\AppBound{\phi}{\cur{D}}-\AppBound{\phi}{\cur{C}}$, $\SEss{\phi}{\cur{D}}-\SEss{\phi}{\cur{C}}$
		and $\SRem{\phi}{\cur{C}}-\SRem{\phi}{\cur{C}}$ can be checked using Corollary \ref{cor.InfChanges}.
		\end{Cor}
		\begin{proof}
			Items \eqref{cor.bclasschange.1}, \eqref{cor.bclasschange.2}, \eqref{cor.bclasschange.3},
			\eqref{cor.bclasschange.4} and \eqref{cor.bclasschange.5} follow directly from Proposition \ref{Prop.Altdefinition}
			and Corollary \ref{cor.InfChanges}. Item \eqref{cor.bclasschange.6} follows from Propositions
			\ref{Prop.Altdefinition} and \ref{prop.InfChanges}. The first half of
			item \eqref{cor.bclasschange.7}
			was proven during the proof of Proposition \ref{Prop.Altdefinition}.
			We include it here for completeness. The second half 
			follows from Definition \ref{def.singpoints}
			and as $A-(B- C)=(A\cap C)\cup(A- B)$
			and $(A- B)\cap C=A\cap(C- B)$.
			Item \eqref{cor.bclasschange.8} follows from
			Proposition \ref{Prop.Altdefinition},
			as $(A\cap B)-(A\cap C)=A\cap(B- C)$,
			$A-(B\cap C)=(A- B)\cup (A- C)$,
			$(B\cap A)- C=B\cap(A- C)$ and the distribution
			law for $\cap$ and $\cup$. Items \eqref{cor.bclasschange.9} and \eqref{cor.bclasschange.10} follow from
			Proposition \ref{Prop.Altdefinition} and the proof of item \eqref{cor.bclasschange.8}.
			Item \eqref{cor.bclasschange.11} follows from Propositions \ref{Prop.Altdefinition} and \ref{prop.InfChanges} and the
			set relations given above.
		\end{proof}

		Corollary \ref{cor.bclasschange} gives the relationships between
		the sets making up the boundary classifications with respect to 
		$\cur{C}$ and $\cur{D}$. We can rewrite the results above to emphasise
		the behaviour of individual boundary points. Since this is the form of the
		results that is likely to be the most useful, we denote it
		as a theorem.
		
		\begin{Thm}\label{boundaryPointChanges}
			Let $\cur{C},\cur{D}\in\BPPSet$ so that $\cur{C}\subset_{\textnormal{b.p.p.}}\cur{D}$ then for all $\phi\in\Env$ we have the following results.
			\begin{enumerate}
				\item If $x\in\App{\phi}{\cur{C}}$ then $x\in\App{\phi}{\cur{D}}$.
				\item If $x\in\Unapp{\phi}{\cur{C}}$ then either $x\in\Unapp{\phi}{\cur{D}}$ or $x\in\App{\phi}{\cur{D}}$.
				\item If $x\in\Inf{\phi}{\cur{C}}$ then
							\begin{enumerate}
								\item $x\in\AppUnbound{\phi}{\cur{D}}$ implies $x\in\Inf{\phi}{\cur{D}}$
								\item $x\not\in\AppUnbound{\phi}{\cur{D}}$ implies $x\in\Sing{\phi}{\cur{D}}$
							\end{enumerate}
				\item If $x\in\RemInf{\phi}{\cur{C}}$ then
							\begin{enumerate}
								\item $x\in\AppUnbound{\phi}{\cur{D}}$ implies $x\in\RemInf{\phi}{\cur{D}}$
								\item $x\not\in\AppUnbound{\phi}{\cur{D}}$ implies $x\in\RemSing{\phi}{\cur{D}}$
							\end{enumerate}
				\item If $x\in\EssInf{\phi}{\cur{C}}$ then
							\begin{enumerate}
								\item $x\in\AppUnbound{\phi}{\cur{D}}$ implies $x\in\EssInf{\phi}{\cur{D}}$
								\item $x\not\in\AppUnbound{\phi}{\cur{D}}$ implies $x\in\Sing{\phi}{\cur{C}}$ and
											\begin{enumerate}
												\item $x\in\SRem{\phi}{\cur{D}}$ implies $x\in\RemSing{\phi}{\cur{D}}$.
												\item $x\not\in\SRem{\phi}{\cur{D}}$ implies $x\in\EssSing{\phi}{\cur{D}}$.
											\end{enumerate}
							\end{enumerate}
				\item If $x\in\MixInf{\phi}{\cur{C}}$ then
							\begin{enumerate}
								\item $x\in\AppUnbound{\phi}{\cur{D}}$ implies $x\in\MixInf{\phi}{\cur{D}}$.
								\item $x\not\in\AppUnbound{\phi}{\cur{D}}$ implies $x\in\Sing{\phi}{\cur{D}}$ and
											\begin{enumerate}
												\item $x\in\SRem{\phi}{\cur{D}}$ implies $x\in\RemSing{\phi}{\cur{D}}$.
												\item $x\not\in\SRem{\phi}{\cur{D}}$ implies $x\in\MixSing{\phi}{\cur{D}}$.
											\end{enumerate}
							\end{enumerate}
				\item If $x\in\PureInf{\phi}{\cur{C}}$ then
							\begin{enumerate}
								\item $x\in\AppUnbound{\phi}{\cur{D}}$ implies $x\in\PureInf{\phi}{\cur{D}}$.
								\item $x\not\in\AppUnbound{\phi}{\cur{D}}$ implies $x\in\Sing{\phi}{\cur{D}}$ and
											\begin{enumerate}
												\item $x\in\SRem{\phi}{\cur{D}}$ implies $x\in\RemSing{\phi}{\cur{D}}$.
												\item $x\not\in\SRem{\phi}{\cur{D}}$ implies $x\in\PureSing{\phi}{\cur{D}}$.
											\end{enumerate}
							\end{enumerate}
				\item If $x\in\Sing{\phi}{\cur{C}}$ then $x\in\Sing{\phi}{\cur{D}}$.
				\item If $x\in\NSing{\phi}{\cur{C}}$ then
							\begin{enumerate}
								\item $x\in\Reg{\phi}$ implies $x\in\NSing{\phi}{\cur{D}}$.
								\item $x\in\AppUnbound{\phi}{\cur{D}}\cup\Unapp{\phi}{\cur{D}}$ implies $x\in\NSing{\phi}{\cur{D}}$.
								\item $x\not\in\AppUnbound{\phi}{\cur{D}}\cup\Unapp{\phi}{\cur{D}}$ implies $x\in\Sing{\phi}{\cur{D}}$.
							\end{enumerate}
				\item If $x\in\RemSing{\phi}{\cur{C}}$ then 
							\begin{enumerate}
								\item $x\in\SRem{\phi}{\cur{D}}$ implies $x\in\RemSing{\phi}{\cur{D}}$.
								\item $x\not\in\SRem{\phi}{\cur{D}}$ implies $x\in\EssSing{\phi}{\cur{D}}$.
							\end{enumerate}
				\item If $x\in\EssSing{\phi}{\cur{C}}$ then $x\in\EssSing{\phi}{\cur{D}}$.
				\item If $x\in\MixSing{\phi}{\cur{C}}$ then $x\in\MixSing{\phi}{\cur{D}}$.
				\item If $x\in\PureSing{\phi}{\cur{C}}$ then $x\in\PureSing{\phi}{\cur{D}}$.
			\end{enumerate}
		\end{Thm}
		\begin{proof}
			The proofs follow from Propositions \ref{Prop.Altdefinition} and
			\ref{prop.InfChanges}, Corollaries \ref{cor.InfChanges}
			and
			\ref{cor.bclasschange}, the set equations
			given in the proof of Corollary \ref{cor.bclasschange} and the
			relevant definitions.
		\end{proof}
		
		Figure \ref{figure bounary}, 
		gives a graphic representation of Theorem \ref{boundaryPointChanges}.  The solid arrows give the usual structure of the boundary classification, the 
		dashed arrows show how the classification may change when $\cur{C}$ is enlarged to $\cur{D}$, i.e. $\cur{C}\subset_{\textnormal{b.p.p.}}\cur{D}$. 
		Note that
		we have not included arrows based at a class that point to the same class.
		
		\begin{figure*}
		  \centering
		  \def\svgwidth{\columnwidth}
		  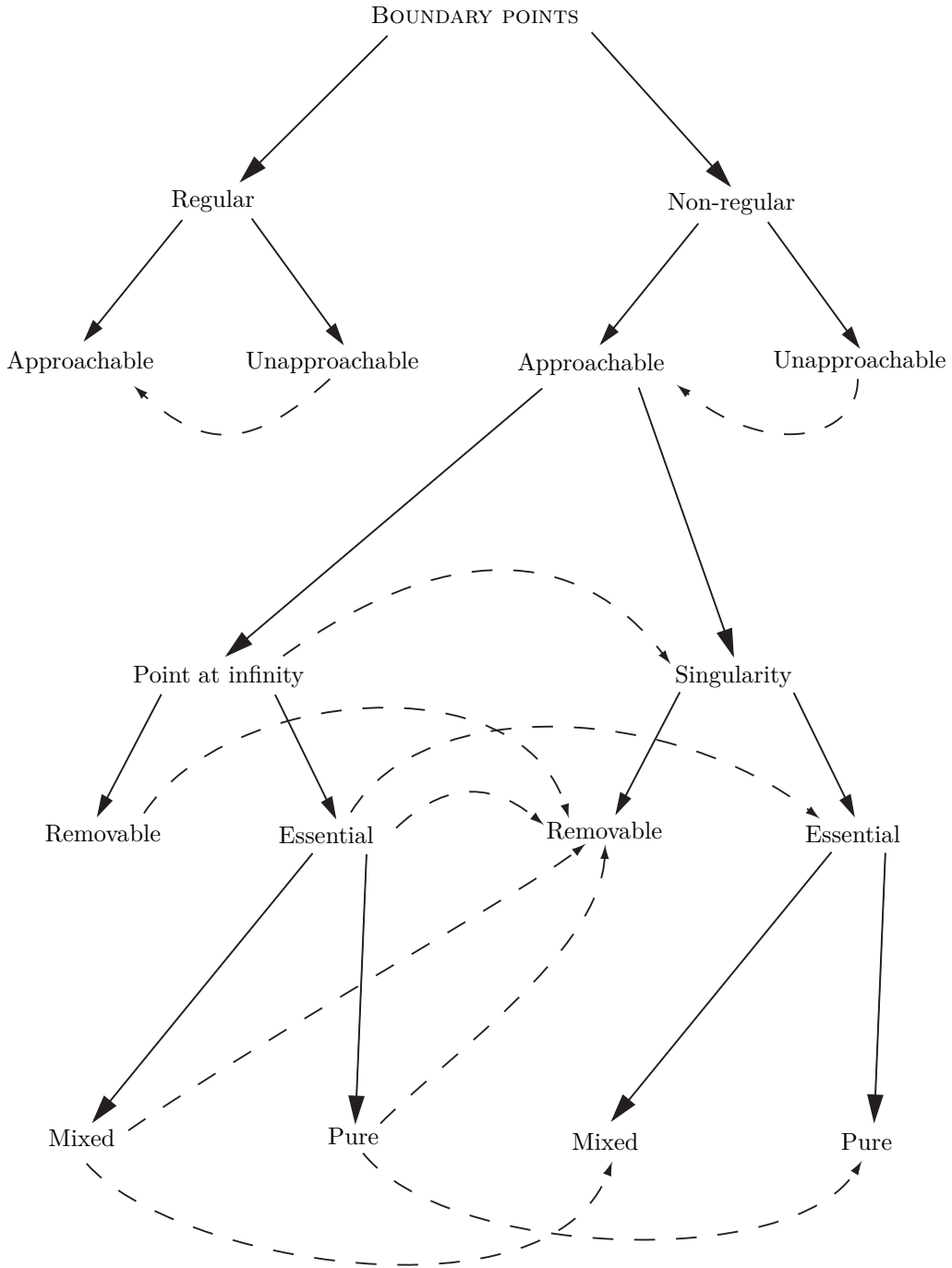
		  \caption{The changes to the boundary classification when \ensuremath{\cur{C}\subset_{\textnormal{b.p.p.}}\cur{D}}. The solid arrows give the usual structure of the boundary classification, the 
		dashed arrows show how the classification may change when $\cur{C}$ is enlarged to $\cur{D}$. Note that
		we have not included arrows based at a class that point to the same class.}\label{figure bounary}
		\end{figure*}
		
		The only surprising behaviour is that an essential point
		at infinity can become a removable singularity. This can
		only occur in a very specific set of circumstances, related to 
		the difference between $\IRem{\phi}$ and $\SRem{\phi}{\cur{C}}$.
		\begin{Lem}\label{lem.wierdbehaviour}
			Let $\cur{C},\cur{D}\in\BPPSet$ be such that
			$\cur{C}\subset_{\textnormal{b.p.p.}}\cur{D}$. 
			If $x\in\EssInf{\phi}{\cur{C}}$ 
			is such that $x\in\RemSing{\phi}{\cur{D}}$ then
			for all $(\psi,U)\in\Bound{M}$ so that $(\psi,U)\covers (\phi,\{x\})$
			we know that,
			\begin{enumerate}
				\item there exists $p\in U$ so that $p\in\Irreg{\psi}\cap\bigl(\AppUnbound{\psi}{\cur{D}}\cup\Unapp{\psi}{\cur{D}}\bigr)$,\label{lem.wierdbehaviour.1}
				\item for all $\gamma\in \cur{D}$, bounded, so that $x\in\overline{\phi\circ\gamma}$
			if $q\in U$ is such that $q\in\overline{\psi\circ\gamma}$
			then $q\in\Reg{\psi}$,\label{lem.wierdbehaviour.2}
				\item $\Reg{\psi}\cap U\neq\EmptySet$.\label{lem.wierdbehaviour.3}
			\end{enumerate}
		\end{Lem}
		\begin{proof}
			Since
			$x\in\RemSing{\phi}{\cur{D}}=\Irreg{\phi}\cap\AppBound{\phi}{\cur{D}}\cap\SRem{\phi}{\cur{D}}$
			we know that there exists $(\psi, U)\in\Bound{M}$ so that
			$(\psi,U)\covers (\phi,\{x\})$ and $U\subset\NSing{\psi}{\cur{D}}$.
			If $U\subset\Reg{\psi}$ then $x\in\NSing{\psi}{\cur{C}}$. This is a contradiction and therefore
			there exists $p\in U$ so that $p\in\Irreg{\psi}$. As 
			$p\in\NSing{\psi}{\cur{D}}$ we know that
			either $p\in\AppUnbound{\psi}{\cur{D}}$ or $p\in\Unapp{\psi}{\cur{D}}$. This proves item
			\eqref{lem.wierdbehaviour.1}.

			Since $x\in\AppBound{\phi}{\cur{D}}$ there exists $\gamma\in \cur D$,
			bounded, so that $x\in\overline{\phi\circ\gamma}$. By Theorem 17
			of \cite{ScottSzekeres1994} there exists $q\in U$ so that
			$q\in\overline{\psi\circ\gamma}$. If $q\in\Irreg{\psi}$
			then $q\in\Sing{\psi}{\cur{D}}$. This is a contradiction
			and therefore $q\in\Reg{\psi}$. This proves items \eqref{lem.wierdbehaviour.2}
			and \eqref{lem.wierdbehaviour.3}.
		\end{proof}

		Thus in order for a boundary point $x\in\partial\phi(M)$, as given in 
		the statement of Lemma \ref{lem.wierdbehaviour},
		to exist, we know that, using the parametrization of curves given by $\cur{D}$, there exists at least one
		\begin{enumerate}
			\item bounded curve, converging to $x$, along which
				the metric behaves regularly. That is, the curves affine parameter
				is bounded, the Kretschmann scalar has a well defined 
				limit, etc...
			\item unbounded curve, converging to $x$,
		\end{enumerate}
		and the set of curves $\cur{C}$ cannot contain 
		any of the curves of the first type and
		must contain at least one curve of the second type
		while the set of curves $\cur{D}$ must contain curves of both types.
		It is tempting therefore to attribute the existence of such points, in the
		classifications of $\cur{C}$ and $\cur{D}$,
		to a surfeit of curves in the set $\cur{C}$. 
		An example of this behaviour can be constructed from the
		directional singularity of the Curzon solution as it hides a portion
		of spacelike infinity, \cite{ScottSzekeres1986b,ScottSzekeres1986a}.
		
		\subsection{Union}\label{Union subsection}
		
			We now give a similar analysis for the union of $\cur{C},\cur{D}\in\BPPSet$.
			Throughout this section we assume that $\cur{C}\cup_{\textnormal{b.p.p.}}\cur{D}$ is 
			well defined.
			
			We will not go into as much detail as Section \ref{subsetsClass} for two reasons.
			First, the need to determine how the classification changes when taking the
			union of two b.p.p.\ satisfying sets is much less than that of adding additional
			curves to a b.p.p.\ satisfying set. Second, the first section provides an
			adequate example of how to determine additional detail if required.
		
			\begin{Prop}\label{prop.UniChanges}
			  Let $\cur{C},\cur{D}\in\BPPSet$ and let
			  $\cur{E}=\cur{C}\cup_{\textnormal{b.p.p.}}\cur{D}$ then, for all $\phi\in\Env$,
			  \begin{enumerate}
			    \item $\App{\phi}{\cur{C}}\cup\App{\phi}{\cur{D}}=\App{\phi}{\cur{E}}$,\label{prop.unichanges.1}
			    \item $\AppBound{\phi}{\cur{C}}\cup\AppBound{\phi}{\cur{D}}=\AppBound{\phi}{\cur{E}}$,\label{prop.unichanges.2}
			    \item $\AppUnbound{\phi}{\cur{E}}=\Bigl(\AppUnbound{\phi}{\cur{C}}-\AppBound{\phi}{\cur{D}}\Bigr)\cup\Bigl(\AppUnbound{\phi}{\cur{D}}-\AppBound{\phi}{\cur{C}}\Bigr)$,\label{prop.unichanges.3}
			    \item $\Unapp{\phi}{\cur{E}}=\Unapp{\phi}{\cur{C}}\cap\Unapp{\phi}{\cur{D}}$,\label{prop.unichanges.4}
			    \item $\NSing{\phi}{\cur{E}}=\NSing{\phi}{\cur{C}}\cap\NSing{\phi}{\cur{D}}$,\label{prop.unichanges.5}
			    \item $\SRem{\phi}{\cur{E}}\subset\SRem{\phi}{\cur{C}}\cap\SRem{\phi}{\cur{D}}$,\label{prop.unichanges.6}
					\item If $x\in\bigl(\SRem{\phi}{\cur{C}}\cap\SRem{\phi}{\cur{D}}\bigr)- \SRem{\phi}{\cur{E}}$
						then for 
						for all $(\psi,U)\in\Bound{M}$ so that $(\psi,U)\covers (\phi,\{x\})$ and
						$U\subset\NSing{\psi}{\cur{C}}$ ($U\subset\NSing{\psi}{\cur{D}}$) there exists $y\in U$ so that
						$y\in\Irreg{\psi}$ and either $y\in\AppUnbound{\psi}{\cur{C}}\cap\AppBound{\psi}{\cur{D}}$
						($y\in\AppUnbound{\psi}{\cur{D}}\cap\AppBound{\psi}{\cur{C}}$)
						or $y\in\Unapp{\psi}{\cur{C}}\cap\App{\psi}{\cur{D}}$ ($y\in\Unapp{\psi}{\cur{D}}\cap\App{\psi}{\cur{C}}$),\label{prop.unichanges.7}
					\item $\SEss{\phi}{\cur{C}}\cup\SEss{\phi}{\cur{D}}\subset\SEss{\phi}{\cur{E}}$,\label{prop.unichanges.8}
					\item $\SEss{\phi}{\cur{E}}-\bigl(\SEss{\phi}{\cur{C}}\cup\SEss{\phi}{\cur{D}}\bigr)=\bigl(\SRem{\phi}{\cur{C}}\cap\SRem{\phi}{\cur{D}}\bigr)- \SRem{\phi}{\cur{E}}$.\label{prop.unichanges.9}
			  \end{enumerate}
			\end{Prop}
			\begin{proof}
				Items \eqref{prop.unichanges.1}, \eqref{prop.unichanges.2}, follow from the construction of $\cur{E}$, see Definition 31 of
				\cite{Whale2012a}. Item \eqref{prop.unichanges.3} follows from items \eqref{prop.unichanges.1} and \eqref{prop.unichanges.2},
			  the equation $(A\cup B)-C=(A-C)\cup(B-C)$ and as $B\subset A$ implies that $A-(B\cup C)=A-B-C$. Item \eqref{prop.unichanges.4}
			  follows from item \eqref{prop.unichanges.1} and the definition.
				Item \eqref{prop.unichanges.5} follows from item \eqref{prop.unichanges.2}, Definition \ref{def.singpoints} and the
				equation $A-(B\cup C)=(A-B)\cap(A-C)$.
			  Item \eqref{prop.unichanges.6} follows directly from item \eqref{prop.unichanges.5} and Definition \ref{def.srem-sess}.
			  
			  We now prove Item \eqref{prop.unichanges.7}. Let $p\in\bigl(\SRem{\phi}{\cur{C}}\cap\SRem{\phi}{\cur{D}}\bigr)- \SRem{\phi}{\cur{E}}$
			  then for all $(\psi,U)\in\Bound{M}$ so that $(\psi,U)\covers (\phi,\{p\})$
			  and $U\subset \NSing{\phi}{\cur{C}}$ we know that $U\not\subset\NSing{\psi}{\cur{E}}$.
			  Thus there exists $q\in U$ so that $q\in\Irreg{\psi}$,
			  $q\not\in\AppUnbound{\psi}{\cur{E}}$
			  and $q\not\in\Unapp{\psi}{\cur{E}}$.
			  Since $q\in\NSing{\psi}{\cur{C}}$ we know that
			  $q\in\AppUnbound{\psi}{\cur{C}}\cup\Unapp{\psi}{\cur{C}}$.
			  If $q\in\AppUnbound{\psi}{\cur{C}}$ then $q\in\AppBound{\psi}{\cur{D}}$ (by item \eqref{prop.unichanges.3} and assumption),
			  otherwise $q\in\NSing{\psi}{\cur{E}}$.
			  If $q\in\Unapp{\psi}{\cur{C}}$ then $q\in\App{\psi}{\cur{D}}$ (by item \eqref{prop.unichanges.4} and assumption),
			  otherwise $q\in\NSing{\psi}{\cur{E}}$. This is sufficient to 
			  prove item \eqref{prop.unichanges.7}.
			  
			  Items \eqref{prop.unichanges.8} and \eqref{prop.unichanges.9} follow from item \eqref{prop.unichanges.6} and the relevant 
			  definitions.
			\end{proof}

		Figures \ref{figure appinfunion} and \ref{figure remsingunion} give a graphical representation of Proposition
		\ref{prop.UniChanges}. We skip the equivalent of Corollary \ref{cor.bclasschange}, the missing details
		can be determined from Proposition \ref{prop.UniChanges}.
		
%		{\smash{\ensuremath{\partial\phi(\man{M})}}}}%
%    {\smash{\ensuremath{\App{\phi}{\cur{C}}}}}}%
%    {\smash{\ensuremath{\App{\phi}{\cur{D}}}}}}%
%    {\smash{\ensuremath{\AppBound{\phi}{\cur{C}}}}}}%
%    {\smash{\ensuremath{\Irreg{\phi}}}}}%
%    {\smash{\ensuremath{\AppBound{\phi}{\cur{D}}}}}}%

		\begin{figure*}
		  \centering
		  \def\svgwidth{\columnwidth}
		  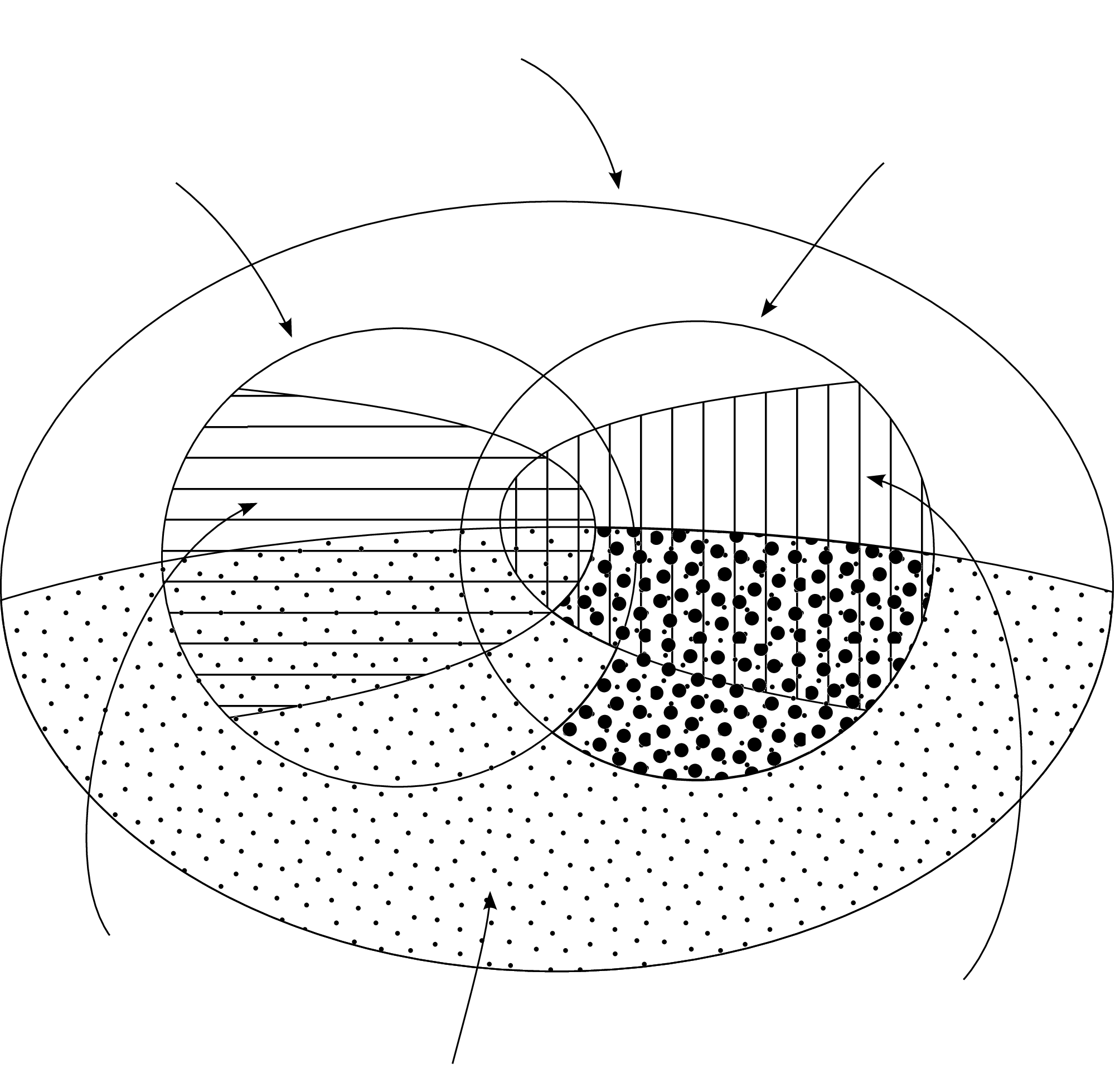
		  \caption{A venn diagram of 
		    $\partial\phi(\man{M})$, $\Reg{\phi}$, $\Irreg{\phi}$,
		    $\App{\phi}{\cur{C}}$, $\AppBound{\phi}{\cur{C}}$,
		    $\AppUnbound{\phi}{\cur{C}}$, $\App{\phi}{\cur{D}}$,
		    $\AppBound{\phi}{\cur{D}}$, $\AppUnbound{\phi}{\cur{D}}$,
		    $\App{\phi}{\cur{C}\cup_{\textnormal{b.p.p.}}\cur{D}}$,
		    $\AppBound{\phi}{\cur{C}\cup_{\textnormal{b.p.p.}}\cur{D}}$ and 
		    $\AppUnbound{\phi}{\cur{C}\cup_{\textnormal{b.p.p.}}\cur{D}}$.
		    The largest oval is $\partial\phi(\man{M})$. The dotted region,
		    both small and large dots, is $\Irreg{\phi}$.
		    The left circle is $\App{\phi}{\cur{C}}$, 
		    the right circle is $\App{\phi}{\cur{D}}$. 
		    The horizontally ruled region is $\AppBound{\phi}{\cur{C}}$, 
		    the region of the left circle not ruled by horizontal
		    lines is $\AppUnbound{\phi}{\cur{C}}$. 
		    Likewise, the vertically ruled region is $\AppBound{\phi}{\cur{D}}$
		    and the region of the right circle not ruled by vertical lines is 
		    $\AppUnbound{\phi}{\cur{D}}$. 
		    The union of both circles is $\App{\phi}{\cur{C}\cup_{\textnormal{b.p.p.}}\cur{D}}$.
		    The region that is ruled, either by horizontal lines or vertical lines 
		    or both is $\AppBound{\phi}{\cur{C}\cup_{\textnormal{b.p.p.}}\cur{D}}$. The region 
		    inside the union of the circles 
		    that is not ruled is $\AppUnbound{\phi}{\cur{C}\cup_{\textnormal{b.p.p.}}\cur{D}}$.
		    The region covered by large dots is 
		    $\Irreg{\phi}\cap\Unapp{\phi}{\cur{C}}
		    \cap\App{\phi}{\cur{D}}$ union $\Irreg{\phi}\cap
		    \AppUnbound{\phi}{\cur{C}}\cap\AppBound{\phi}{\cur{D}}$, see
		    item \eqref{prop.unichanges.7} of Proposition 
		    \ref{prop.UniChanges}. Note that the sets mentioned in brackets in
		    item \eqref{prop.unichanges.7} of Proposition 
		    \ref{prop.UniChanges} can be found by interchanging $\cur{C}$ and $\cur{D}$
		    in the labels of this diagram.}\label{figure appinfunion}
		\end{figure*}
		
    \begin{figure*}
		  \centering
		  \def\svgwidth{\columnwidth}
		  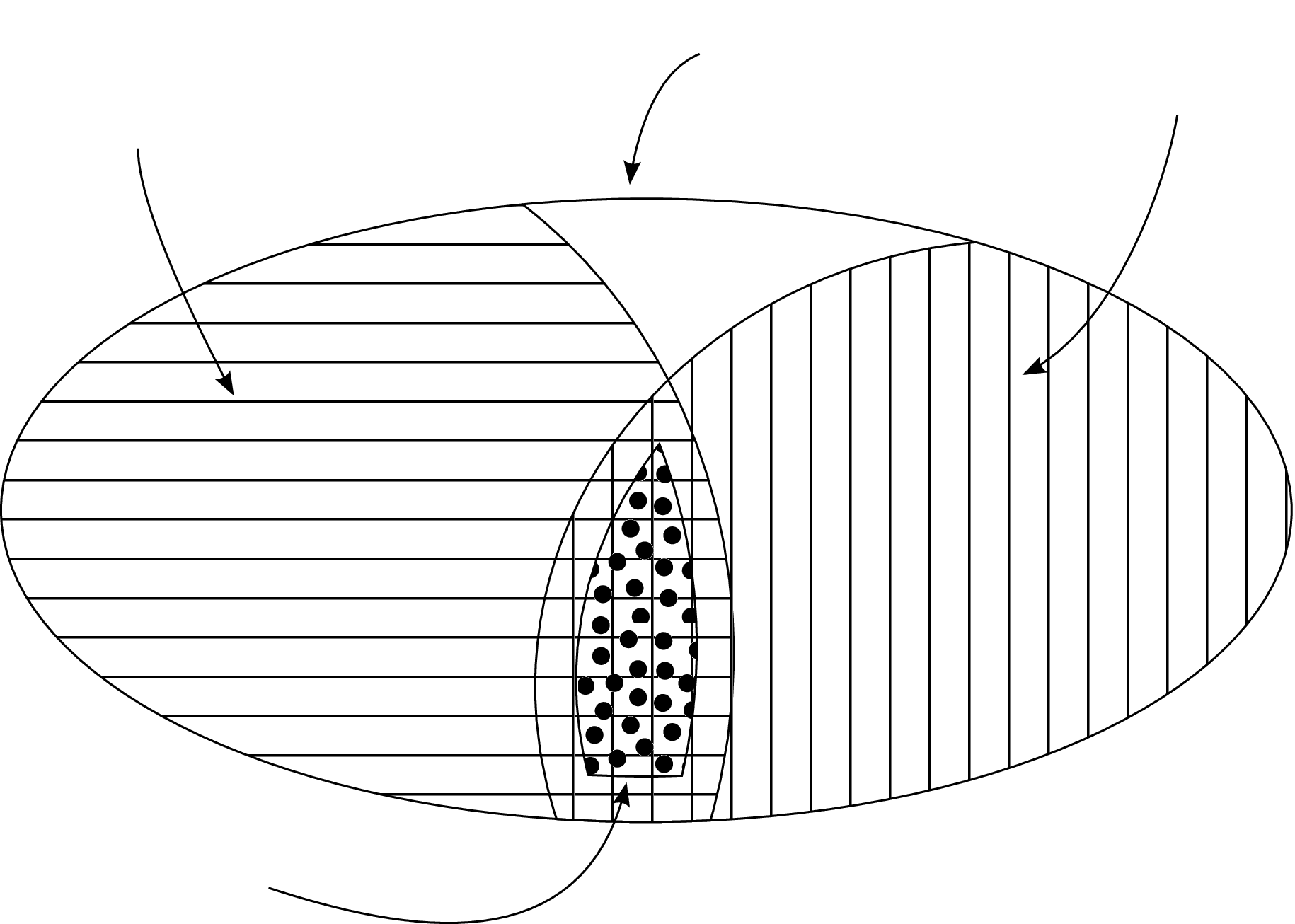
		  \caption{A venn diagram of 
		    $\partial\phi(\man{M})$, $\SRem{\phi}{\cur{C}}$, $\SEss{\phi}{\cur{C}}, 
		    \SRem{\phi}{\cur{D}}$, $\SEss{\phi}{\cur{D}}$, $\SRem{\phi}{\cur{C}\cup_{\textnormal{b.p.p.}}\cur{D}}$
		    and $\SEss{\phi}{\cur{C}\cup_{\textnormal{b.p.p.}}\cur{D}}$. The oval is
		    $\partial\phi(\man{M})$. The horizontally ruled region is $\SRem{\phi}{\cur{C}}$,
		    the region not ruled by horizontal lines is $\SEss{\phi}{\cur{C}}$.
		    Likewise, the vertically ruled region is $\SRem{\phi}{\cur{D}}$ and
		    the region not ruled by vertical lines is $\SEss{\phi}{\cur{D}}$.
		    The dotted region is $\SRem{\phi}{\cur{C}\cup_{\textnormal{b.p.p.}}\cur{D}}$ and
		    the region that is not dotted is $\SEss{\phi}{\cur{C}\cup_{\textnormal{b.p.p.}}\cur{D}}$.
		    The region that is ruled by both horizontal and vertical lines, but
		    is not dotted consists of those boundary points that are described
		    by item \eqref{prop.unichanges.7} of Proposition \ref{prop.UniChanges},
		    see Figure \ref{figure appinfunion}.}\label{figure remsingunion}
		\end{figure*}
		
%{\smash{\ensuremath{\partial\phi(\man{M})}}}}%
%{\smash{\ensuremath{\SRem{\phi}{\cur{C}}}}}}%
%{\smash{\ensuremath{\SRem{\phi}{\cur{D}}}}}}%
%{\smash{\ensuremath{\SRem{\phi}{\cur{C}\cup_{\textnormal b.p.p.}\cur{D}}}}}}%

		\begin{Thm}\label{thm.union}
		  Let $\cur{C},\cur{D}\in\BPPSet$ and let
		  $\cur{E}=\cur{C}\cup_{\textnormal{b.p.p.}}\cur{D}$. For all $\phi\in\Env$ we have that,
			\begin{enumerate}
			  \item If $x\in\App{\phi}{\cur{C}}$ then $x\in\App{\phi}{\cur{E}}$
				\item If $x\in\Unapp{\phi}{\cur{C}}$ then
					\begin{enumerate}
					  \item $x\in\App{\phi}{\cur{D}}$ implies $x\in\App{\phi}{\cur{E}}$
					  \item $x\in\Unapp{\phi}{\cur{D}}$ implies $x\in\Unapp{\phi}{\cur{E}}$
					\end{enumerate}
				\item If $x\in\Inf{\phi}{\cur{C}}$ then
					\begin{enumerate}
					  \item $x\not\in\AppBound{\phi}{\cur{D}}$ implies $x\in\Inf{\phi}{\cur{E}}$
					  \item $x\in\AppBound{\phi}{\cur{D}}$ implies $x\in\Sing{\phi}{\cur{E}}$
					\end{enumerate}
				\item If $x\in\RemInf{\phi}{\cur{C}}$ then
					\begin{enumerate}
					  \item $x\not\in\AppBound{\phi}{\cur{D}}$ implies $x\in\RemInf{\phi}{\cur{E}}$
					  \item $x\in\AppBound{\phi}{\cur{D}}$ implies $x\in\RemSing{\phi}{\cur{E}}$
					\end{enumerate}
				\item If $x\in\EssInf{\phi}{\cur{C}}$ then
					\begin{enumerate}
					  \item $x\not\in\AppBound{\phi}{\cur{D}}$ implies $x\in\EssInf{\phi}{\cur{E}}$
					  \item $x\in\AppBound{\phi}{\cur{D}}$ implies $x\in\Sing{\phi}{\cur{E}}$ and
							\begin{enumerate}
							  \item $x\in\SRem{\phi}{\cur{E}}$ implies $x\in\RemSing{\phi}{\cur{E}}$
							  \item $x\not\in\SRem{\phi}{\cur{E}}$ implies $x\in\EssSing{\phi}{\cur{E}}$
							\end{enumerate}
					\end{enumerate}
				\item If $x\in\MixInf{\phi}{\cur{C}}$ then
					\begin{enumerate}
					  \item $x\not\in\AppBound{\phi}{\cur{D}}$ implies $x\in\MixInf{\phi}{\cur{E}}$
					  \item $x\in\AppBound{\phi}{\cur{D}}$ implies $x\in\Sing{\phi}{\cur{E}}$ and
							\begin{enumerate}
							  \item $x\in\SRem{\phi}{\cur{E}}$ implies $x\in\RemSing{\phi}{\cur{E}}$
							  \item $x\not\in\SRem{\phi}{\cur{E}}$ implies $x\in\MixSing{\phi}{\cur{E}}$
							\end{enumerate}
					\end{enumerate}
				\item If $x\in\PureInf{\phi}{\cur{C}}$ then
					\begin{enumerate}
					  \item $x\not\in\AppBound{\phi}{\cur{D}}$ implies $x\in\PureInf{\phi}{\cur{E}}$
					  \item $x\in\AppBound{\phi}{\cur{D}}$ implies $x\in\Sing{\phi}{\cur{E}}$ and
							\begin{enumerate}
							  \item $x\in\SRem{\phi}{\cur{E}}$ implies $x\in\RemSing{\phi}{\cur{E}}$
							  \item $x\not\in\SRem{\phi}{\cur{E}}$ implies $x\in\PureSing{\phi}{\cur{E}}$
							\end{enumerate}
					\end{enumerate}
				\item If $x\in\Sing{\phi}{\cur{C}}$ then $x\in\Sing{\phi}{\cur{E}}$
				\item if $x\in\NSing{\phi}{\cur{C}}$ then
					\begin{enumerate}
					  \item $x\in\Reg{\phi}$ implies $x\in\NSing{\phi}{\cur{E}}$
					  \item $x\not\in\Reg{\phi}$ and
							\begin{enumerate}
							  \item $x\in\AppUnbound{\phi}{\cur{D}}$ implies $x\in\NSing{\phi}{\cur{E}}$
							  \item $x\not\in\AppUnbound{\phi}{\cur{D}}$ implies $x\in\Sing{\phi}{\cur{E}}$
							\end{enumerate}
					\end{enumerate}
					\item If $x\in\RemSing{\phi}{\cur{C}}$ then
						\begin{enumerate}
						  \item $x\in\SRem{\phi}{\cur{E}}$ implies $x\in\RemSing{\phi}{\cur{E}}$
						  \item $x\not\in\SRem{\phi}{\cur{E}}$ implies $x\in\EssSing{\phi}{\cur{E}}$
						\end{enumerate}
					\item If $x\in\EssSing{\phi}{\cur{C}}$ then $x\in\EssSing{\phi}{\cur{E}}$
					\item If $x\in\MixSing{\phi}{\cur{C}}$ then $x\in\MixSing{\phi}{\cur{E}}$
					\item If $x\in\PureSing{\phi}{\cur{C}}$ them $x\in\PureSing{\phi}{\cur{E}}$.
			\end{enumerate}
			Membership of $\SRem{\phi}{\cur{E}}$ can be checked using item \eqref{prop.unichanges.7}
			of Proposition \ref{prop.UniChanges}.  The same statements hold when interchanging
			$\cur{C}$ and $\cur{D}$.
		\end{Thm}
		\begin{proof}
		  Each item follows from Propositions \ref{Prop.Altdefinition} and \ref{prop.UniChanges}.
		\end{proof}

		We again see the surprising behaviour that essential points at infinity
		can become removable singularities. As before, this
		can only occur in a specific set of circumstances, namely those described in 
		item \eqref{prop.unichanges.7} of Proposition \ref{prop.UniChanges} and by Lemma 
		\ref{lem.wierdbehaviour} (where we consider $\cur{C}\subset_{\textnormal{b.p.p.}}\cur{E}$).
		
		An example of this, fitting the situation of
		item \eqref{prop.unichanges.7} of
		Proposition \ref{prop.UniChanges}, can be constructed using the
		directional singularity of the Curzon solution, \cite{ScottSzekeres1986b,ScottSzekeres1986a},
		where $\cur{C}$ contains curves classifying the directional singularity
		as a point of spacelike infinity and $\cur{D}$ contains curves
		classifying the directional singularity as a singularity.

\section{Changes to the abstract boundary classification as we change $\cur{C}$}\label{sec.abstratboundarychanges}

	Now that we know how the classification of boundary points changes
	we can determine how the classification of abstract boundary points changes.

	\subsection{Subset}
		
			\begin{Thm}\label{thm.abchangessubset}
				Let $\cur{C},\cur{D}\in\BPPSet$ so that $\cur{C}\subset_{\textnormal{b.p.p.}}\cur{D}$ then,
				\begin{enumerate}
					\item If $[(\phi,\{x\})]\in\ABApp{\cur{C}}$ then $[(\phi,\{x\})]\in\ABApp{\cur{D}}$
					\item If $[(\phi,\{x\})]\in\ABUnapp{\cur{C}}$ then either $[(\phi,\{x\})]\in\ABApp{\cur{D}}$ or $[(\phi,\{x\})]\in\ABUnapp{\cur{D}}$.
					\item If $[(\phi,\{x\})]\in\ABReg{\cur{C}}$ then 
						\begin{enumerate}
							\item $x\in\Reg{\phi}$ implies $[(\phi,\{x\})]\in\ABReg{\cur{D}}$
							\item $x\in\AppUnbound{\phi}{\cur{C}}$ implies $[(\phi,\{x\})]\in\ABReg{\cur{D}}$
							\item $x\in\AppBound{\phi}{\cur{C}}$ and
							\begin{enumerate}
							  \item $x\in\SRem{\phi}{\cur{D}}$ implies $[(\phi,\{x\})]\in\ABReg{\cur{D}}$
							  \item $x\not\in\SRem{\phi}{\cur{D}}$ implies $[(\phi,\{x\})]\in\ABSing{\cur{D}}$.
							\end{enumerate}
						\end{enumerate}
					\item If $[(\phi,\{x\})]\in\ABInf{\cur{C}}$ then 
						\begin{enumerate}
							\item $x\in\AppUnbound{\phi}{\cur{D}}$ implies $[(\phi,\{x\})]\in\ABInf{\cur{D}}$
							\item $x\not\in\AppUnbound{\phi}{\cur{D}}$ and 
							\begin{enumerate}
							  \item $x\in\SRem{\phi}{\cur{D}}$ implies $[(\phi,\{x\})]\in\ABReg{\cur{D}}$
							  \item $x\not\in\SRem{\phi}{\cur{D}}$ implies $[(\phi,\{x\})]\in\ABSing{\cur{D}}$.
							\end{enumerate}
						\end{enumerate}
					\item If $[(\phi,\{x\})]\in\ABMixInf{\cur{C}}$ then 
						\begin{enumerate}
							\item $x\in\AppUnbound{\phi}{\cur{D}}$ implies $[(\phi,\{x\})]\in\ABMixInf{\cur{D}}$
							\item $x\not\in\AppUnbound{\phi}{\cur{D}}$ and 
							\begin{enumerate}
							  \item $x\in\SRem{\phi}{\cur{D}}$ then $[(\phi,\{x\})]\in\ABReg{\cur{D}}$
							  \item $x\not\in\SRem{\phi}{\cur{D}}$ then $[(\phi,\{x\})]\in\ABMixSing{\cur{D}}$.
							\end{enumerate}
						\end{enumerate}
					\item If $[(\phi,\{x\})]\in\ABPureInf{\cur{C}}$ then 
						\begin{enumerate}
							\item $x\in\AppUnbound{\phi}{\cur{D}}$ implies $[(\phi,\{x\})]\in\ABPureInf{\cur{D}}$
							\item $x\not\in\AppUnbound{\phi}{\cur{D}}$ 
							\begin{enumerate}
							  \item $x\in\SRem{\phi}{\cur{D}}$ then $[(\phi,\{x\})]\in\ABReg{\cur{D}}$
							  \item $x\not\in\SRem{\phi}{\cur{D}}$ then $[(\phi,\{x\})]\in\ABPureSing{\cur{D}}$.
							\end{enumerate}
						\end{enumerate}
					\item If $[(\phi,\{x\})]\in\ABSing{\cur{C}}$ then $[(\phi,\{x\})]\in\ABSing{\cur{D}}$
					\item If $[(\phi,\{x\})]\in\ABMixSing{\cur{C}}$ then $[(\phi,\{x\})]\in\ABMixSing{\cur{D}}$
					\item If $[(\phi,\{x\})]\in\ABPureSing{\cur{C}}$ then $[(\phi,\{x\})]\in\ABPureSing{\cur{D}}$
				\end{enumerate}
			\end{Thm}
			\begin{proof}
				This follows from Definitions \ref{App&UnAppABpoints}, \ref{IndABPOints}, \ref{InfAbPoints} and \ref{SingAbPoints} and from Theorem
				\ref{boundaryPointChanges}. Note that $\Reg{\phi},\Inf{\phi}{\cur{C}}\subset\SRem{\phi}{\cur{C}}$ since elements
			of $\Reg{\phi}$ and $\Inf{\phi}{\cur{C}}$ are covered by themselves.
			\end{proof}
			
			As before we give a graphical depiction of these results in Figure \ref{figure abstract boundar}. 
			The solid arrows give the usual structure of the abstract boundary classification, the 
			dashed arrows show how the classification may change when $\cur{C}$ is enlarged to $\cur{D}$, i.e. $\cur{C}\subset_{\textnormal{b.p.p.}}\cur{D}$.
			Note that
			we have not included arrows based at a class that point to the same class.
			\begin{figure*}
				\centering
				\def\svgwidth{\columnwidth}
				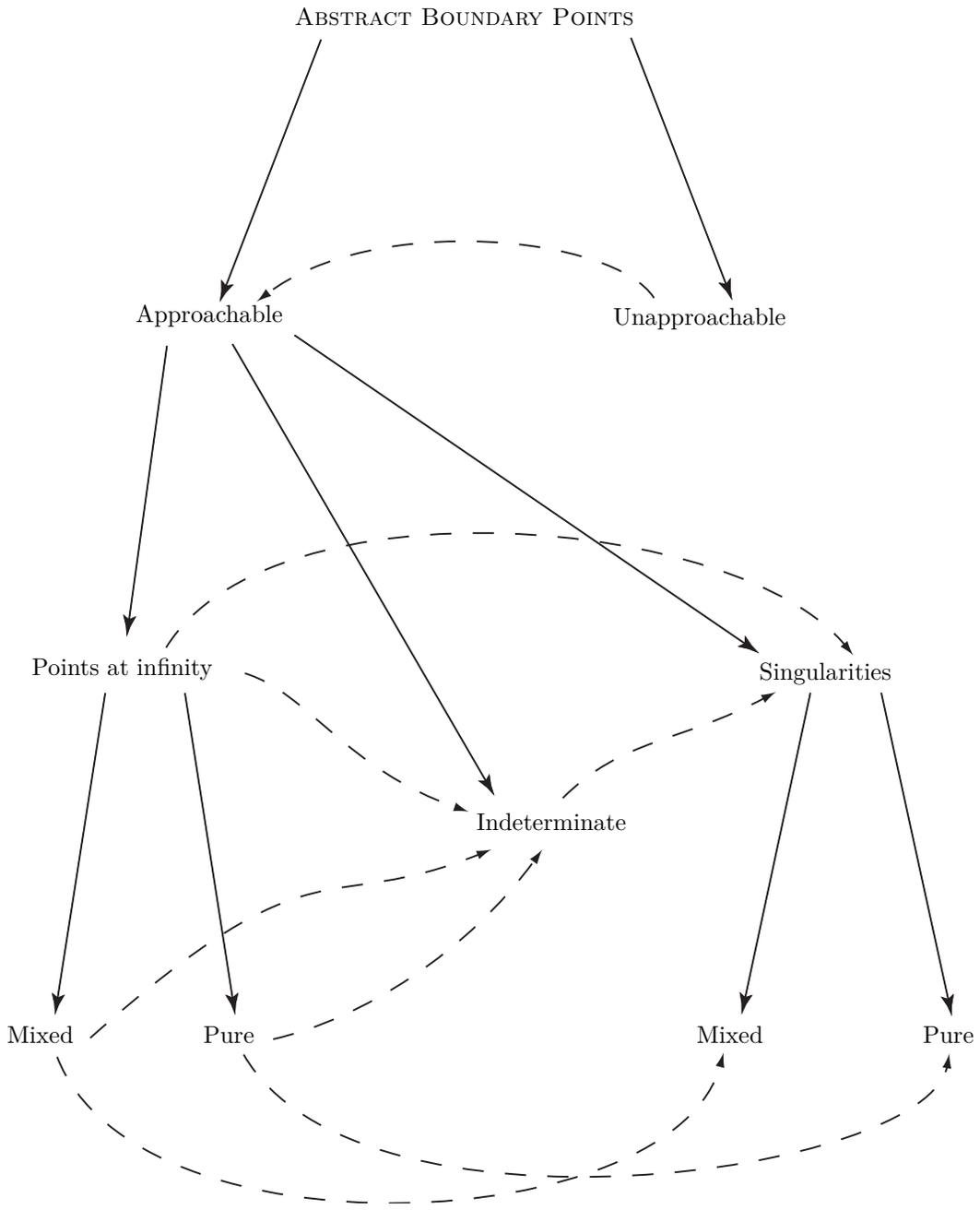 
				\caption{The changes to the abstract boundary classification when $\ensuremath{\cur{C}\subset_{\textnormal{b.p.p.}}\cur{D}}$. The solid arrows give the usual structure of the boundary classification, the 
		dashed arrows show how the classification may change between $\cur{C}$ and $\cur{D}$. Note that
		we have not included arrows based at a class that point to the same class.}\label{figure abstract boundar}
			\end{figure*}
			
		\subsubsection{Union}
		
			\begin{Thm}
				Let $\cur{C},\cur{D}\in\BPPSet$ so that $\cur{E}=\cur{C}\cup_{\textnormal{b.p.p.}}\cur{D}$ is well defined, then,
				\begin{enumerate}
					\item If $[(\phi,\{x\})]\in\ABApp{\cur{C}}$ then $[(\phi,\{x\})]\in\ABApp{\cur{E}}$.
					\item If $[(\phi,\{x\})]\in\ABUnapp{\cur{C}}$ then
						\begin{enumerate}
						  \item $[(\phi,\{x\})]\in\ABApp{\cur{D}}$ implies  $[(\phi,\{x\})]\in\ABApp{\cur{E}}$
						  \item $[(\phi,\{x\})]\in\ABUnapp{\cur{D}}$ implies  $[(\phi,\{x\})]\in\ABUnapp{\cur{E}}$.
						\end{enumerate}
					\item If $[(\phi,\{x\})]\in\ABReg{\cur{C}}$ then
						\begin{enumerate}
							\item $x\in\Reg{\phi}$ implies $[(\phi,\{x\})]\in\ABReg{\cur{E}}$
							\item $x\in\AppUnbound{\phi}{\cur{C}}$ implies $[(\phi,\{x\})]\in\ABReg{\cur{E}}$
							\item $x\in\AppBound{\phi}{\cur{C}}$ and
							\begin{enumerate}
							  \item $x\in\SRem{\phi}{\cur{E}}$ implies $[(\phi,\{x\})]\in\ABReg{\cur{E}}$
							  \item $x\not\in\SRem{\phi}{\cur{E}}$ implies $[(\phi,\{x\})]\in\ABSing{\cur{E}}$.
							\end{enumerate}
						\end{enumerate}
					\item If $[(\phi,\{x\})]\in\ABInf{\cur{C}}$ then 
						\begin{enumerate}
							\item $x\not\in\AppBound{\phi}{\cur{D}}$ implies $[(\phi,\{x\})]\in\ABInf{\cur{E}}$
							\item $x\in\AppBound{\phi}{\cur{D}}$ and 
							\begin{enumerate}
							  \item $x\in\SRem{\phi}{\cur{E}}$ implies $[(\phi,\{x\})]\in\ABReg{\cur{E}}$
							  \item $x\not\in\SRem{\phi}{\cur{E}}$ implies $[(\phi,\{x\})]\in\ABSing{\cur{E}}$.
							\end{enumerate}
						\end{enumerate}
					\item If $[(\phi,\{x\})]\in\ABMixInf{\cur{C}}$ then 
						\begin{enumerate}
							\item $x\not\in\AppBound{\phi}{\cur{D}}$ implies $[(\phi,\{x\})]\in\ABMixInf{\cur{E}}$
							\item $x\in\AppBound{\phi}{\cur{D}}$ and
							\begin{enumerate}
							  \item $x\in\SRem{\phi}{\cur{E}}$ implies $[(\phi,\{x\})]\in\ABReg{\cur{E}}$
							  \item $x\not\in\SRem{\phi}{\cur{E}}$ implies $[(\phi,\{x\})]\in\ABMixSing{\cur{E}}$.
							\end{enumerate}
						\end{enumerate}
					\item If $[(\phi,\{x\})]\in\ABPureInf{\cur{C}}$ then 
						\begin{enumerate}
							\item $x\not\in\AppBound{\phi}{\cur{D}}$ implies $[(\phi,\{x\})]\in\ABPureInf{\cur{E}}$
							\item $x\in\AppBound{\phi}{\cur{D}}$ and 
							\begin{enumerate}
							  \item $x\in\SRem{\phi}{\cur{E}}$ implies $[(\phi,\{x\})]\in\ABReg{\cur{E}}$
							  \item $x\not\in\SRem{\phi}{\cur{E}}$ implies $[(\phi,\{x\})]\in\ABPureSing{\cur{E}}$.
							\end{enumerate}
						\end{enumerate}
					\item If $[(\phi,\{x\})]\in\ABSing{\cur{C}}$ then $[(\phi,\{x\})]\in\ABSing{\cur{E}}$
					\item If $[(\phi,\{x\})]\in\ABMixSing{\cur{C}}$ then $[(\phi,\{x\})]\in\ABMixSing{\cur{E}}$
					\item If $[(\phi,\{x\})]\in\ABPureSing{\cur{C}}$ then $[(\phi,\{x\})]\in\ABPureSing{\cur{E}}$.
				\end{enumerate}
				Membership of $\SRem{\phi}{\cur{E}}$ can be checked using item \eqref{prop.unichanges.7}
				of Proposition \ref{prop.UniChanges}. The same statements hold when interchanging
				$\cur{C}$ and $\cur{D}$.
			\end{Thm}
			\begin{proof}
				This follows from Definitions \ref{App&UnAppABpoints}, \ref{IndABPOints}, \ref{InfAbPoints} and \ref{SingAbPoints} and from Theorem
				\ref{thm.union}.
			\end{proof}

\section{Acknowledgements}
The author was partially funded by Marsden grant UOO-09-022.

\bibliographystyle{unsrt}
\bibliography{AB_classification_II_Ben_Whale}

\end{document}

%% file: figure1.eps_tex
%% Creator: Inkscape inkscape 0.48.1, www.inkscape.org
%% Accompanies image file 'Appinf.ps' (pdf, eps, ps)

\begingroup
  \makeatletter
  \providecommand\color[2][]{%
    \errmessage{(Inkscape) Color is used for the text in Inkscape, but the package 'color.sty' is not loaded}
    \renewcommand\color[2][]{}%
  }
  \providecommand\transparent[1]{%
    \errmessage{(Inkscape) Transparency is used (non-zero) for the text in Inkscape, but the package 'transparent.sty' is not loaded}
    \renewcommand\transparent[1]{}%
  }
  \providecommand\rotatebox[2]{#2}
  \ifx\svgwidth\undefined
    \setlength{\unitlength}{579.425pt}
  \else
    \setlength{\unitlength}{\svgwidth}
  \fi
  \global\let\svgwidth\undefined
  \makeatother
  \begin{picture}(1,0.67556868)%
    \put(0,0){\includegraphics[width=\unitlength]{figure1.eps}}%
    \put(0.16714836,0.58216454){\color[rgb]{0,0,0}\makebox(0,0)[lb]{\smash{\ensuremath{\AppUnbound{\phi}{\cur{D}}}}}}%
    \put(0.82948107,0.58524931){\color[rgb]{0,0,0}\makebox(0,0)[lb]{\smash{\ensuremath{\AppBound{\phi}{\cur{D}}}}}}%
    \put(0.10955919,0.0548841){\color[rgb]{0,0,0}\makebox(0,0)[lb]{\smash{\ensuremath{\AppUnbound{\phi}{\cur{C}}}}}}%
    \put(0.41855733,0.00027371){\color[rgb]{0,0,0}\makebox(0,0)[lb]{\smash{\ensuremath{\AppBound{\phi}{\cur{C}}}}}}%
    \put(0.77897951,0.02492107){\color[rgb]{0,0,0}\makebox(0,0)[lb]{\smash{\ensuremath{\Irreg{\phi}}}}}%
    \put(0.35945118,0.66628954){\color[rgb]{0,0,0}\makebox(0,0)[lb]{\smash{\ensuremath{\partial\phi(\man{M})}}}}%
  \end{picture}%
\endgroup

%% file: figure2.eps_tex
%% Creator: Inkscape inkscape 0.48.1, www.inkscape.org
%% Accompanies image file 'figure2.ps' (pdf, eps, ps)

\begingroup
  \makeatletter
  \providecommand\color[2][]{%
    \errmessage{(Inkscape) Color is used for the text in Inkscape, but the package 'color.sty' is not loaded}
    \renewcommand\color[2][]{}%
  }
  \providecommand\transparent[1]{%
    \errmessage{(Inkscape) Transparency is used (non-zero) for the text in Inkscape, but the package 'transparent.sty' is not loaded}
    \renewcommand\transparent[1]{}%
  }
  \providecommand\rotatebox[2]{#2}
  \ifx\svgwidth\undefined
    \setlength{\unitlength}{495.35pt}
  \else
    \setlength{\unitlength}{\svgwidth}
  \fi
  \global\let\svgwidth\undefined
  \makeatother
  \begin{picture}(1,0.64784108)%
    \put(0,0){\includegraphics[width=\unitlength]{figure2.eps}}%
    \put(0.03685049,0.6313581){\color[rgb]{0,0,0}\makebox(0,0)[lb]{\smash{\ensuremath{\partial\phi(\man{M})}}}}%
    \put(0.83314556,0.6134125){\color[rgb]{0,0,0}\makebox(0,0)[lb]{\smash{\ensuremath{\SRem{\phi}{\cur{D}}}}}}%
    \put(0.53876924,0){\color[rgb]{0,0,0}\makebox(0,0)[lb]{\smash{\ensuremath{\SRem{\phi}{\cur{C}}}}}}%
  \end{picture}%
\endgroup

%% file: figure3.eps_tex
%% Creator: Inkscape inkscape 0.48.1, www.inkscape.org
%% Accompanies image file 'BoundaryClassificationDiagramCurveClasses.ps' (pdf, eps, ps)

\begingroup
  \makeatletter
  \providecommand\color[2][]{%
    \errmessage{(Inkscape) Color is used for the text in Inkscape, but the package 'color.sty' is not loaded}
    \renewcommand\color[2][]{}%
  }
  \providecommand\transparent[1]{%
    \errmessage{(Inkscape) Transparency is used (non-zero) for the text in Inkscape, but the package 'transparent.sty' is not loaded}
    \renewcommand\transparent[1]{}%
  }
  \providecommand\rotatebox[2]{#2}
  \ifx\svgwidth\undefined
    \setlength{\unitlength}{557.66127184pt}
  \else
    \setlength{\unitlength}{\svgwidth}
  \fi
  \global\let\svgwidth\undefined
  \makeatother
  \begin{picture}(1,1.30042729)%
    \put(0,0){\includegraphics[width=\unitlength]{figure3.eps}}%
    \put(0.68861733,0.60175792){\color[rgb]{0,0,0}\makebox(0,0)[lb]{\smash{Singularity}}}%
    \put(0.37492634,1.28126243){\color[rgb]{0,0,0}\makebox(0,0)[lb]{\smash{\sc Boundary points}}}%
    \put(0.16930566,1.09237828){\color[rgb]{0,0,0}\makebox(0,0)[lb]{\smash{Regular}}}%
    \put(0.68096645,1.08879188){\color[rgb]{0,0,0}\makebox(0,0)[lb]{\smash{Non-regular}}}%
    \put(0.52603348,0.92333911){\color[rgb]{0,0,0}\makebox(0,0)[lb]{\smash{Approachable}}}%
    \put(0.78975393,0.92620819){\color[rgb]{0,0,0}\makebox(0,0)[lb]{\smash{Unapproachable}}}%
    \put(0.12985521,0.60151883){\color[rgb]{0,0,0}\makebox(0,0)[lb]{\smash{Point at infinity}}}%
    \put(0.03899958,0.43773958){\color[rgb]{0,0,0}\makebox(0,0)[lb]{\smash{Removable}}}%
    \put(0.28000611,0.43558772){\color[rgb]{0,0,0}\makebox(0,0)[lb]{\smash{Essential}}}%
    \put(0.042586,0.1233313){\color[rgb]{0,0,0}\makebox(0,0)[lb]{\smash{Mixed}}}%
    \put(0.33093307,0.12572219){\color[rgb]{0,0,0}\makebox(0,0)[lb]{\smash{Pure}}}%
    \put(0.24596721,0.92506172){\color[rgb]{0,0,0}\makebox(0,0)[lb]{\smash{Unapproachable}}}%
    \put(-0.00020174,0.92506172){\color[rgb]{0,0,0}\makebox(0,0)[lb]{\smash{Approachable}}}%
    \put(0.55595482,0.44030681){\color[rgb]{0,0,0}\makebox(0,0)[lb]{\smash{Removable}}}%
    \put(0.8226486,0.4367204){\color[rgb]{0,0,0}\makebox(0,0)[lb]{\smash{Essential}}}%
    \put(0.58244824,0.11980166){\color[rgb]{0,0,0}\makebox(0,0)[lb]{\smash{Mixed}}}%
    \put(0.85979705,0.11980159){\color[rgb]{0,0,0}\makebox(0,0)[lb]{\smash{Pure}}}%
  \end{picture}%
\endgroup

%% file: figure4.eps_tex
%% Creator: Inkscape inkscape 0.48.1, www.inkscape.org
%% Accompanies image file 'AppinfUnion.ps' (pdf, eps, ps)

\begingroup
  \makeatletter
  \providecommand\color[2][]{%
    \errmessage{(Inkscape) Color is used for the text in Inkscape, but the package 'color.sty' is not loaded}
    \renewcommand\color[2][]{}%
  }
  \providecommand\transparent[1]{%
    \errmessage{(Inkscape) Transparency is used (non-zero) for the text in Inkscape, but the package 'transparent.sty' is not loaded}
    \renewcommand\transparent[1]{}%
  }
  \providecommand\rotatebox[2]{#2}
  \ifx\svgwidth\undefined
    \setlength{\unitlength}{575.375pt}
  \else
    \setlength{\unitlength}{\svgwidth}
  \fi
  \global\let\svgwidth\undefined
  \makeatother
  \begin{picture}(1,0.97879571)%
    \put(0,0){\includegraphics[width=\unitlength]{figure4.eps}}%
    \put(0.40986763,0.93751828){\color[rgb]{0,0,0}\makebox(0,0)[lb]{\smash{\ensuremath{\partial\phi(\man{M})}}}}%
    \put(0.1103363,0.82827275){\color[rgb]{0,0,0}\makebox(0,0)[lb]{\smash{\ensuremath{\App{\phi}{\cur{C}}}}}}%
    \put(0.78964479,0.84376575){\color[rgb]{0,0,0}\makebox(0,0)[lb]{\smash{\ensuremath{\App{\phi}{\cur{D}}}}}}%
    \put(0.09524056,0.10963587){\color[rgb]{0,0,0}\makebox(0,0)[lb]{\smash{\ensuremath{\AppBound{\phi}{\cur{C}}}}}}%
    \put(0.38086795,0.00078753){\color[rgb]{0,0,0}\makebox(0,0)[lb]{\smash{\ensuremath{\Irreg{\phi}}}}}%
    \put(0.8265897,0.07626633){\color[rgb]{0,0,0}\makebox(0,0)[lb]{\smash{\ensuremath{\AppBound{\phi}{\cur{D}}}}}}%
  \end{picture}%
\endgroup

%% file: figure5.eps_tex
%% Creator: Inkscape inkscape 0.48.2, www.inkscape.org
%% Accompanies image file 'RemSingUnion.ps' (pdf, eps, ps)
\begingroup%
  \makeatletter%
  \providecommand\color[2][]{%
    \errmessage{(Inkscape) Color is used for the text in Inkscape, but the package 'color.sty' is not loaded}%
    \renewcommand\color[2][]{}%
  }%
  \providecommand\transparent[1]{%
    \errmessage{(Inkscape) Transparency is used (non-zero) for the text in Inkscape, but the package 'transparent.sty' is not loaded}%
    \renewcommand\transparent[1]{}%
  }%
  \providecommand\rotatebox[2]{#2}%
  \ifx\svgwidth\undefined%
    \setlength{\unitlength}{526.075bp}%
    \ifx\svgscale\undefined%
      \relax%
    \else%
      \setlength{\unitlength}{\unitlength * \real{\svgscale}}%
    \fi%
  \else%
    \setlength{\unitlength}{\svgwidth}%
  \fi%
  \global\let\svgwidth\undefined%
  \global\let\svgscale\undefined%
  \makeatother%
  \begin{picture}(1,0.71369324)%
    \put(0,0){\includegraphics[width=\unitlength]{figure5.eps}}%
    \put(0.55193704,0.66934951){\color[rgb]{0,0,0}\makebox(0,0)[lb]{\smash{\ensuremath{\partial\phi(\man{M})}}}}%
    \put(0.09816289,0.6086112){\color[rgb]{0,0,0}\makebox(0,0)[lb]{\smash{\ensuremath{\SRem{\phi}{\cur{C}}}}}}%
    \put(0.89090252,0.63871945){\color[rgb]{0,0,0}\makebox(0,0)[lb]{\smash{\ensuremath{\SRem{\phi}{\cur{D}}}}}}%
    \put(0.11973207,0.04346493){\color[rgb]{0,0,0}\makebox(0,0)[lb]{\smash{\ensuremath{\SRem{\phi}{\cur{C}\cup_{\textnormal b.p.p.}\cur{D}}}}}}%
  \end{picture}%
\endgroup%

%% file: figure6.eps_tex
%% Creator: Inkscape inkscape 0.48.1, www.inkscape.org
%% Accompanies image file 'AbstractBoundaryClassification.ps' (pdf, eps, ps)

\begingroup
  \makeatletter
  \providecommand\color[2][]{%
    \errmessage{(Inkscape) Color is used for the text in Inkscape, but the package 'color.sty' is not loaded}
    \renewcommand\color[2][]{}%
  }
  \providecommand\transparent[1]{%
    \errmessage{(Inkscape) Transparency is used (non-zero) for the text in Inkscape, but the package 'transparent.sty' is not loaded}
    \renewcommand\transparent[1]{}%
  }
  \providecommand\rotatebox[2]{#2}
  \ifx\svgwidth\undefined
    \setlength{\unitlength}{546.12465459pt}
  \else
    \setlength{\unitlength}{\svgwidth}
  \fi
  \global\let\svgwidth\undefined
  \makeatother
  \begin{picture}(1,1.23720444)%
    \put(0,0){\includegraphics[width=\unitlength]{figure6.eps}}%
    \put(0.29390391,1.21763473){\color[rgb]{0,0,0}\makebox(0,0)[lb]{\smash{\sc Abstract Boundary Points}}}%
    \put(0.13035127,0.91039186){\color[rgb]{0,0,0}\makebox(0,0)[lb]{\smash{Approachable}}}%
    \put(0.62295866,0.90812949){\color[rgb]{0,0,0}\makebox(0,0)[lb]{\smash{Unapproachable}}}%
    \put(0.02317744,0.54589694){\color[rgb]{0,0,0}\makebox(0,0)[lb]{\smash{Points at infinity}}}%
    \put(0.77380676,0.54124949){\color[rgb]{0,0,0}\makebox(0,0)[lb]{\smash{Singularities}}}%
    \put(0.48179991,0.38615203){\color[rgb]{0,0,0}\makebox(0,0)[lb]{\smash{Indeterminate}}}%
    \put(-0.00258784,0.16628803){\color[rgb]{0,0,0}\makebox(0,0)[lb]{\smash{Mixed}}}%
    \put(0.20009929,0.16603554){\color[rgb]{0,0,0}\makebox(0,0)[lb]{\smash{Pure}}}%
    \put(0.70908872,0.16590545){\color[rgb]{0,0,0}\makebox(0,0)[lb]{\smash{Mixed}}}%
    \put(0.94269423,0.16590539){\color[rgb]{0,0,0}\makebox(0,0)[lb]{\smash{Pure}}}%
  \end{picture}%
\endgroup